\documentclass[10pt,oneside]{article}
\usepackage{graphics}
\usepackage{amsmath}
\usepackage{amssymb}
\usepackage{enumerate}
\usepackage{mathrsfs}
\usepackage[dvips]{graphicx}
\usepackage{psfrag}
\newtheorem{thm}{Theorem}
\newtheorem{lem}{Lemma}
\newtheorem{coro}{Corollary}

\newenvironment{proof}{\par\indent{\emph{Proof.}}}{\hfill$\square$\vspace{1.4mm}\par\noindent}

\newcommand{\X}{\mathbf{x}}

\newcommand{\Y}{\mathbf{y}}

\newcommand{\Dn}[2]{\frac{\partial #1}{\partial n_{#2}}}
\newcommand{\Dxi}[2]{\frac{\partial #1}{\partial x_{#2}}}

\newcommand{\Oj}{\mathbf{O}}
\newcommand{\bZ}{\mathbf{z}}
\newcommand{\bW}{\mathbf{w}}

\newcommand{\cD}{\mathcal{D}}

\newcommand{\N}{\boldsymbol{n}}

\numberwithin{equation}{section}
\DeclareMathOperator*{\diag}{diag}

\title{Mesoscale asymptotic approximations to solutions of mixed boundary value problems in perforated domains}
\author{V. Maz'ya\footnote{Department of Mathematical Sciences, University of Liverpool, Liverpool L69 3BX, U.K., and Department of Mathematics, Link\"oping University, SE-581 83 Link\"oping, Sweden.}, A. Movchan\footnote{Department of Mathematical Sciences, University of Liverpool, Liverpool L69 3BX, U.K.}, M. Nieves$^\dagger$}
\date{}

\newcommand\eq[1] {(\ref{#1})}

\newcommand{\bfm}[1]{\mbox{\boldmath ${#1}$}}

\newcommand{\beqa}{\begin{eqnarray}}
\newcommand{\eeqa}{\end{eqnarray}}
\newcommand{\bequ}{\begin{equation}}
\newcommand{\eequ}[1]{\label{#1}\end{equation}}
\newcommand{\Grad}{\nabla}

\newcommand{\ov}[1]{\overline{#1}}

\newcommand{\Ga}{\alpha}

\newcommand{\Gd}{\delta}

\newcommand{\Gve}{\varepsilon}

\newcommand{\Gvf}{\varphi}

\newcommand{\Gl}{\lambda}

\newcommand{\Go}{\omega}

\newcommand{\GD}{\Delta}

\newcommand{\GS}{\Sigma}

\newcommand{\GO}{\Omega}
\newcommand{\GX}{\Xi}

\newcommand{\BGT}{\bfm\Theta}

\newcommand{\BGX}{\bfm\Xi}

\newcommand{\CK}{{\cal K}}
\newcommand{\CL}{{\cal L}}

\newcommand{\CQ}{{\cal Q}}
\newcommand{\CR}{{\cal R}}
\newcommand{\CS}{{\cal S}}

\newcommand{\CV}{{\cal V}}

\newcommand{\BCC}{{\bfm{\cal C}}}
\newcommand{\BCD}{{\bfm{\cal D}}}

\newcommand{\BCN}{{\bfm{\cal N}}}

\newcommand{\BCQ}{{\bfm{\cal Q}}}

\newcommand{\BCS}{{\bfm{\cal S}}}

\def\Bw{{\bf w}}
\def\Bx{{\bf x}}
\def\By{{\bf y}}
\def\Bz{{\bf z}}

\def\BC{{\bf C}}

\def\BO{{\bf O}}

\def\BQ{{\bf Q}}

\def\BX{{\bf X}}
\def\BY{{\bf Y}}

\newcommand{\beq}{\begin{equation}}
\newcommand{\eeq}{\end{equation}}
\newcommand{\overliner}{\begin{eqnarray}}
\newcommand{\earr}{\end{eqnarray}}
\newcommand{\beqn}{\begin{equation*}}
\newcommand{\eeqn}{\end{equation*}}
\newcommand{\overlinern}{\begin{eqnarray*}}
\newcommand{\earrn}{\end{eqnarray*}}
\newcommand{\prt}{\partial}

\newcommand{\fr}{\frac}

\def\l{\label}

\begin{document}
\maketitle
\begin{abstract}
We describe a  method of  
asymptotic approximations to solutions of  mixed boundary value problems for the Laplacian in a 
three-dimensional domain with many perforations of arbitrary shape, with the Neumann boundary conditions being prescribed on the surfaces of small voids.
The only assumption made on the geometry is that  
the diameter of a void is assumed to be smaller compared to the distance to the  
nearest neighbour. 
The asymptotic approximation,
obtained here, involves a linear combination of dipole fields constructed for individual voids, with the coefficients, which are determined by solving a linear algebraic system.
We prove the solvability of this system and derive an estimate for its solution. The energy estimate is obtained for the remainder term of the asymptotic approximation.
\end{abstract}
\section{Introduction}\label{Introductionmeso}

In the present paper we discuss a method for  asymptotic approximations to solutions of  the mixed problems for the Poisson equation for  domains  containing  a number, possibly large,  of  
small perforations of arbitrary shape. The Dirichlet condition is set on the exterior boundary of the perforated body, and the Neumann conditions are specified  on the boundaries of small holes. 
Neither periodicity nor even local ``almost'' periodicity constraints  are imposed on the position of holes,  which makes the homogenization methodologies not applicable (cf.  Chapter 4 in \cite{MKh}, and Chapter 5 in \cite{SP}).  
Two geometrical parameters, $\varepsilon$ and $d$, are 
introduced to characterize the maximum diameter of perforations within the array and the minimum distance between the voids, respectively. Subject to 
the 
{\em mesoscale} constraint, $\varepsilon < \mbox{const} ~d$, the 
asymptotic approximation to the solution of the mixed boundary value problem is constructed. 
The approximate  
solution involves a linear combination of dipole fields constructed for individual voids, with the coefficients
 determined 
 from a 
 linear algebraic system. 
The formal asymptotic 
representation  is accompanied by the energy estimate of the remainder term.
The general idea of mesoscale approximations originated a couple of years ago in \cite{MMMeso}, where the Dirichlet 
problem was considered for a domain with multiple inclusions.

The asymptotic methods, presented here and in  \cite{MMMeso},  can be applied to modelling of dilute composites in problems of mechanics, electromagnetism, heat conduction and phase transition.  
%
 %
In such models, the boundary conditions have to be satisfied across a large array of small voids, which is the situation fully served by our approach. Being 
used in the case of a dilute array of small spherical particles, the method also includes the physical models of many point interactions treated previously in \cite{Figari, FigPapRub, Honig} and elsewhere.
%
Asymptotic approximations applied to solutions of boundary value problems of mixed type in domains containing many small spherical inclusions were considered in \cite{Figari}. The point interaction approximations to solutions of diffusion problems in domains with many small spherical holes were analysed in \cite{FigPapRub}.   Modelling of multi-particle interaction  in problems of phase transition was considered in \cite{Honig} where the evolution of a large number of small 
spherical particles embedded into an ambient medium takes place during the last stage of phase transformation; such a phenomenon where particles in a melt are subjected to growth is referred to as Ostwald ripening. For the numerical treatment of models involving large number $N$ of spherical particles, the fast multipole method, of order $O(N)$,  was proposed    in \cite{Greengard_Rokhlin}, and it appears to be efficient 
 for the rapid evaluation of potential and force fields for systems of a large number of particles  interacting with each other via the Coulomb law.

 We give an outline of the paper. 
 The notation $\Omega_N$ will be used for a 
 domain containing 
 small voids $F^{(j)}$, $j=1, \dots, N$, while the 
 unperturbed   
 domain, without any holes, is denoted by $\GO$.  The number $N$ is assumed to be large.
 
 If $\GO$ is a bounded domain in ${\Bbb R}^3$, we introduce $L^{1,2}(\GO)$ as the space of functions on $\GO$ with
 distributional first derivatives in $L^2(\GO)$ provided with the norm
 \beq
 \|u\|_{L^{1,2}(\GO)} = \Big(   \|\Grad u \|^2_{L^2(\GO)}  + \| u \|^2_{L^2(\GO)}   \Big)^{1/2}.
 \eequ{norm}
 Here $B$ is a ball at a positive distance from $\prt \GO.$ If $\GO$ is unbounded, by $L^{1,2}(\GO)$ we mean the completion of the space of functions with $\| \Grad u \|_{L^2(\GO)} < \infty,$ which have bounded supports, in the norm \eq{norm}. The space of traces of functions in $L^{1,2}(\GO)$ on $\prt \GO$ will be denoted by $L^{{\tiny 1/2},2}(\prt \GO).$
 
 The maximum of diameters of $F^{(j)}, ~ j=1,\ldots, N$ is denoted by $\Gve$. 
 An array of points $\BO^{(j)}, ~j = 1, \ldots, N, $ is chosen in such a way that
 $\BO^{(j)}$ is an interior point of $F^{(j)}$ for every $j=1,\ldots, N$.
 By $2d$ we denote  the smallest distance between  the points within the array $\{\Oj^{(j)}\}_{j=1}^N$. It is assumed that there exists an open set $\Go \subset \GO$ situated at a positive distance from $\prt \GO$ and such that
 \beq
 \cup_{j=1}^N F^{(j)} \subset \Go, 
~ \mbox{dist} \big(   \cup_{j=1}^N F^{(j)}, \prt \Go \big) \geq 2d,   ~\mbox{and} ~ \mbox{diam  } \Go = 1,
 \eequ{intro_1}
With the last normalization of the size of $\Go$, the parameters $\Gve$ and $d$ 
can be considered as 
 non-dimensional.
The scaled 
open sets $\Gve^{-1} F^{(j)}$ are assumed to have 
Lipschitz boundaries, with Lipschitz characters independent of $N$. 

Our goal is to obtain an asymptotic approximation to 
a unique solution $u_N \in L^{1,2}(\GO_N)$ of the problem
\begin{eqnarray}
-\Delta u_N(\X)=f(\X)\;, \quad \X \in \Omega_N\;,\l{intro_2} \\
u_N(\X)= 
\phi(\Bx)\;, \quad \X \in \partial \Omega\;,\l{intro_3} \\
 \Dn{u_N}{}(\X)=0\;, \quad \X \in \partial F^{(j)} \;, j=1, \dots, N\;,\l{intro_4}
\end{eqnarray}
where $\phi \in L^{1/2, 2}(\prt \GO)$ and $f (\Bx)$ is a 
function in $L^\infty(\GO)$ with 
compact support 
at a positive distance from the cloud 
$\omega$ of small perforations. 

We need solutions to certain model problems in order to construct  the approximation to $u_N$; these include  
\begin{enumerate}
\item $v$ 
as the solution of the unperturbed problem in $\Omega$ (without voids),

\item $\BCD^{(k)}$ as the vector function whose components are the dipole fields for the void $F^{(k)}$,

\item $H$ as the regular part of Green's function $G$ in $\Omega$.
\end{enumerate}

The approximation relies upon a certain algebraic system, incorporating the field $v_f$ and  integral characteristics associated with the small voids. We define
\[\BGT=\left(\Dxi{v}{1}           
(\Oj^{(1)}), \Dxi{v}{2}  
(\Oj^{(1)}),\Dxi{v}{3}  
(\Oj^{(1)}), \dots,\Dxi{v}{1}  
(\Oj^{(N)}), \Dxi{v}{2}  
(\Oj^{(N)}),\Dxi{v}{3}
(\Oj^{(N)})\right)^T\;,\]
and $
{\frak S}=[\mathfrak{S}_{ij}]_{i, j=1}^N$ 
which is a $3N \times 3N$ matrix  with $3\times 3$ block entries 
\[\mathfrak{S}_{ij}=\left\{\begin{array}{ll}\displaystyle{(\nabla_{\bZ} \otimes \nabla_{\bW})\left(G(\Bz, \Bw) 
\right)\Big|_{\substack{\bZ=\Oj^{(i)}\\ \bW=\Oj^{(j)}}}} & \quad \text{ if } i \ne j\\
0I_3&\quad \text{otherwise}\end{array}  \right., \]
where $G$ is Green's function in $\GO$, and $I_3$ is the $3 \times 3$ identity matrix.
We also use the block-diagonal matrix
 \beq \mathbf{Q}=\diag\{\BCQ^{(1)}, \dots, \BCQ^{(N)}\}, \eequ{matrix_Q}
where $\BCQ^{(k)}$ is the so-called  $3\times 3$ polarization matrix for the small void $F^{(k)}$ (see 
\cite{MM_Sob} and Appendix G of \cite{MMP}). 
The shapes of the voids $F^{(j)}, j = 1,\ldots, N,$ are constrained in such a way that the maximal and minimal eigenvalues $\Gl_{max}^{(j)}, ~\Gl_{min}^{(j)}$ of the matrices $-\BCQ^{(j)}$ satisfy the inequalities
\beq
A_1 \Gve^3 > \max_{1 \leq j \leq N} \Gl_{max}^{(j)}, ~~ 
\min_{1 \leq j \leq N} \Gl_{min}^{(j)} > A_2 \Gve^3, 
\eequ{Gl_min} 
where 
$A_1$ and  $A_2$ are positive and independent of $\Gve$.

One of the results, for the case when $\GO={\Bbb R}^3$, $H\equiv 0$, and when \eq{intro_3} is replaced by the condition of  decay 
of $u_N$ at infinity, 
can be formulated as follows

\vspace{0.1in}
\begin{thm}
 \label{thm1_alg_f_inf} 
Let  
\[\varepsilon < c\, d\;,\]
where $c$ is a sufficiently small absolute constant. Then the solution $u_N(\X)$ 
admits the asymptotic representation
\begin{equation}\label{introeq1}
u_N(\X)=v
(\X)+
\sum^N_{k=1} \boldsymbol{C}^{(k)} \cdot 
\BCD^{(k)}(\Bx) 
+\CR_N(\X)\;,
\end{equation}
where $\boldsymbol{C}^{(k)}=(C^{(k)}_1, C^{(k)}_2, C^{(k)}_3)^T$ and  the column vector $\mathbf{C}=(C^{(1)}_1, C^{(1)}_2, C^{(1)}_3, \dots, C^{(N)}_1, C^{(N)}_2, C^{(N)}_3)^T$ 
satisfies the invertible linear
algebraic system 
\beq
(\mathbf{I}+ 
{\frak S} \mathbf{Q})\mathbf{C}  = -\BGT\;.
\eequ{alg_s_intro}
The 
remainder $\CR_N$ satisfies the energy
estimate
\begin{equation}\label{introeq2}
 \| \nabla \CR_N\|^2_{L_2(\Omega_N)} \le \text{\emph{const} } \Big\{ \varepsilon^{11
 }d^{-11
 } + \varepsilon^{5}d^{-3} \Big\} \| \nabla v
 \|^2_{L^2( \GO 
 )}  
  .
 \end{equation}
\end{thm}
We remark that since $\Gve$ and $d$ are non-dimensional parameters,  there is no dimensional mismatch in the right-hand side of \eq{introeq2}.

We now describe the plan of the article. In Section \ref{MainNotmeso}, we introduce the multiply-perforated geometry and
consider the above model problems.
The formal asymptotic algorithm for a cloud of small perforations in the infinite space and the analysis of the algebraic system
\eq{alg_s_intro} are given in Sections \ref{R3alg} and \ref{alg_syst}. Section \ref{R3energyesy} presents the 
proof of Theorem \ref{thm1_alg_f_inf}. 
The problem for a cloud of small perforations in a general 
domain is considered in Section \ref{boundformalg}. 
 Finally, in Section \ref{example} we give an illustrative example accompanied by the numerical simulation.


\section{Main notations and model boundary value problems}\label{MainNotmeso}
Let  $\Omega$ be a  bounded 
domain in $\mathbb{R}^3$ with a smooth boundary $\prt \GO$. We shall also consider the case when $\GO = {\Bbb R}^3 .$  

%
 The perforated domain $\Omega_N$, 
 is given by
 \[\Omega_N=\Omega \backslash \overline{\cup_{j=1}^N F^{(j)}}\;,\]
 where  $F^{(j)}$ are small voids introduced in the previous section. Also in the previous section we introduced the notations $\Gve$ and $d$ for two small parameters, characterizing the maximum of the diameters of
 $F^{(j)}, j=1,\ldots,N,$ and the minimal distance between the small voids, respectively.  
 
 In sections where we are concerned with the energy estimates of the remainders produced by asymptotic approximations we 
 frequently use the obvious estimate
 \begin{equation}\label{sumestvd}
 N \le \text{const }d^{-3}\;.
 \end{equation}
 
 We 
  consider the approximation of  the function 
  $u_N$ which is a variational solution of the mixed problem \eq{intro_2}-\eq{intro_4}.

\label{Modelproblems}
Before constructing the approximation to $u_N$, we introduce model auxiliary functions which the asymptotic scheme relies upon.

\begin{enumerate}
\item \emph{
Solution $v$ 
in the unperturbed domain $\Omega$.} Let $v
\in L^{1,2}(\GO)$ denote a 
unique variational solution of the problem
\begin{eqnarray}
-\Delta v
(\X)&&=f(\X)\;, \quad \X \in \Omega\;,\l{vfom_1} \\
 v
 (\X)&&= \phi(\Bx)
 \;, \quad \X \in \prt \Omega\;.  \l{vfom_2}   \end{eqnarray}

\item \emph{
Regular part of Green's function in $\Omega$.} By $H$ we mean the regular part of Green's function $G$ in $\Omega$ defined by the formula
\beq H(\X, \Y)=(4\pi|\X-\Y|)^{-1}-G(\X, \Y)\;.\eequ{H_def}
Then $H$ is a variational solution of 
\[ \Delta_\X H(\X, \Y)=0\;, \quad \X, \Y \in \Omega\;,\]
\[ H(\X, \Y)=(4\pi|\X-\Y|)^{-1}\;, \quad\X \in \partial \Omega, \Y \in \Omega\;.\]

\item \emph{The dipole fields $\cD^{(j)}_i, ~ i=1, 2, 3,$ associated with the  void $F^{(j)}$.} The vector functions $\BCD^{(j)}=\{ \cD^{(j)}_i\}^3_{i=1}$, 
which are called the dipole fields,  
are variational solutions of the exterior Neumann problems
\begin{equation}\label{modfieldeq1}
\left. \begin{array}{c}
\displaystyle{\Delta \BCD^{(j)}(\Bx)=\Oj\;, \quad \Bx \in 
{\Bbb R}^3 \setminus \bar{F}^{(j)}\;,}\\
\displaystyle{ \Dn{\BCD^{(j)}}{}(\Bx)=\N^{(j)}\;,\quad \Bx \in \partial F^{(j)}\;,}\\
\displaystyle{ \BCD^{(j)}(\Bx)=O(\varepsilon^3 |\Bx-\BO^{(j)}|^{-2}) \quad \text{ as }\quad | \Bx | \to \infty \;,}
\end{array}\right\}
 \end{equation}
where $\N^{(j)}$ is the unit outward normal with respect to $F^{(j)}$.
In the text below we also use 
the 
negative definite polarization matrix $\BCQ^{(j)}=\{\CQ^{(j)}_{ik}\}_{i, k=1}^3$,
as well as the following asymptotic result (see \cite{MM_Sob} and Appendix G in \cite{MMP}), for every  void $F^{(j)}$:

\begin{lem}\label{Djasymp}
For $|\Bx - \BO^{(j)}|>2 \Gve$, the dipole fields admit the asymptotic representation
\begin{equation}\label{modfieldeq2}
\cD^{(j)}_i(\Bx)=\frac{1}{4\pi} \sum^3_{m=1} \CQ^{(j)}_{im} \frac{x_m- O^{(j)}_m}{|\Bx - \BO^{(j)}|^3}+O\left(\Gve^{4} {|\Bx-\BO^{(j)}|^{-3}}\right)\;, \quad i=1,2,3\;.
\end{equation}
\end{lem}

The shapes of the voids $F^{(j)}, j = 1,\ldots, N,$ are constrained in such a way that the maximal and minimal eigenvalues $\Gl_{max}^{(j)}, ~\Gl_{min}^{(j)}$ of the matrices $-\BCQ^{(j)}$ satisfy the inequalities \eq{Gl_min}.

\end{enumerate}

\section{The formal approximation of $u_N$ for the infinite space containing many voids}\label{R3alg}

In this section we deduce formally the uniform asymptotic approximation of $u_N$: 
\[ u_N(\X)\sim v
(\X)+ 
\sum^N_{k=1} \boldsymbol{C}^{(k)} \cdot \BCD^{(k)}(\Bx 
)\;,\] 
for the case $\Omega=\mathbb{R}^3$ and derive an algebraic system for 
the coefficients $\boldsymbol{C}^{(k)}=\{C^{(k)}_i\}^3_{i=1}$, $k=1, \dots, N$.

The function $u_N$ satisfies
\begin{equation}\label{R3formalgeq1}
-\Delta u_N(\X)=f(\X)\;, \quad \X \in \Omega_N\;,
\end{equation}
\begin{equation}\label{R3formalgeq2}
\Dn{u_N}{}(\X)=0\;,\quad \X \in  \partial F
^{(j)}, j=1, \dots, N\;,
\end{equation}
\begin{equation}\label{R3formalgeq3}
u_N(\X)\to 0\;, \quad \text{ as }|\X| \to \infty\;.
\end{equation}
We begin by constructing the 
asymptotic representation for $u_N$ in this way
\begin{equation}\label{R3formalgeq5}
u_N(\X)=v
(\X)+     
\sum^N_{k=1} \boldsymbol{C}^{(k)} \cdot \cD^{(k)}(\Bx)+\mathcal{R}_N(\X)\, 
\end{equation}
where $\CR_N$ is the remainder, and $v
(\X)$ satisfies
\[ -\Delta v
(\X)=f(\X)\;, \quad \X \in \mathbb{R}^3\;,\]
\[ v
(\X) \to 0 \quad \text{ as } \quad |\X| \to \infty\;,\]
and $\cD^{(k)}$ are the dipole fields defined as solutions of problems \eq{modfieldeq1}.
The function $\mathcal{R}_N$ is harmonic in $\Omega_N$ and 
\begin{equation}\label{R3formalgeq8}
\mathcal{R}_N(\X) =O(|\X|^{-1}) \quad \text{ as } |\X| \to \infty\;.
\end{equation}
Placement of (\ref{R3formalgeq5
}) into (\ref{R3formalgeq2
}) together with (\ref{modfieldeq1}) gives the boundary condition on $\prt F^{(j)}$:
\[\Dn{\mathcal{R}_N}{}(\X)=-\N^{(j)} \cdot \Big\{ \nabla v
(\Oj^{(j)})+\boldsymbol{C}^{(j)}+O(\varepsilon)+ 
\sum_{\substack{k \ne j\\ 1 \le k \le N}} \nabla (\boldsymbol{C}^{(k)} \cdot \cD^{(k)}(\Bx))\Big\} .
\]
Now we use (\ref{modfieldeq2}), for $\cD^{(k)}$, $k \ne j$, so that this boundary condition becomes
\begin{eqnarray*}
\Dn{\mathcal{R}_N}{}(\X)&\sim&
-\N^{(j)}\cdot\Big\{ \nabla v
(\Oj^{(j)})+\boldsymbol{C}^{(j)}
+\sum_{\substack{k \ne j \\ 1 \le k \le N}} T(\X, \Oj^{(k)})\BCQ^{(k)} \boldsymbol{C}^{(k)} \Big\}
\;, \quad \X \in \partial F^{(j)}, j=1, \dots, N\;,
\end{eqnarray*}
where
\begin{equation}\label{defT}
T(\X, \Y)=(\nabla_{\bZ} \otimes \nabla_{\bW})\left( \frac{1}{4\pi|\bZ-\bW|}\right)\Big|_{\substack{\bZ=\X\\ \bW=\Y}}\;.
\end{equation}
Finally, Taylor's expansion of $T(\X, \Oj^{(k)})$ about $\X=\Oj^{(j)}, ~ j \neq k, $ leads to 
\begin{eqnarray*}
\Dn{\mathcal{R}_N}{}(\X)&\sim&
-\N^{(j)}\cdot\Big\{ \nabla v
(\Oj^{(j)})+\boldsymbol{C}^{(j)}
+\sum_{\substack{k \ne j \\ 1 \le k \le N}} T(\Oj^{(j)}, \Oj^{(k)})\BCQ^{(k)} \boldsymbol{C}^{(k)} 
\Big\}\;, \quad \X \in \partial 
F^{(j)}, j=1, \dots, N\;.
\end{eqnarray*}
To remove the leading order discrepancy in the above boundary condition, 
we require that the vector coefficients $\BC^{(j)}$ satisfy the algebraic system
\begin{equation}\label{R3formalgeq9}
\nabla v
(\Oj^{(j)})+\boldsymbol{C}^{(j)}+
\sum_{\substack{ k \ne j \\ 1 \le k \le N}} T(\Oj^{(j)}, \Oj^{(k)})\BCQ^{(k)}\boldsymbol{C}^{(k)}=\Oj\;, \quad \text{ for } j=1, \dots, N\;,
\end{equation} 
where the polarization matrices $\BCQ^{(j)}$ 
characterize the geometry of 
$F^{(j)}, ~ j= 1, \ldots, N.$ 
Upon solving the above algebraic system, the formal asymptotic approximation of $u_N$ is complete. The next section addresses the solvability of the system   \eq{R3formalgeq9}, together with 
estimates  for
the vector coefficients $\BC^{(j)}.$


\section{\bf 
Algebraic system in the case $\GO = {\Bbb R}^3$} \label{alg_syst}

The algebraic system for the coefficients $\BC^{(j)}$ 
can be written in the form 
\beq
\BC + 
\BCS \BQ \BC = - \BGT, 
\eequ{alg1}
where 
$$
\BC = ( (\BC^{(1)})^T, \ldots, (\BC^{(N)})^T )^T, ~~ \BGT
 = ((\Grad v
 (\BO^{(1)}))^T, \ldots , (\Grad v
 (\BO^{(N)}))^T )^T,
$$
are vectors of the dimension $3N$, and 
\begin{eqnarray}
\BCS&&=[\CS_{ij}]_{i, j=1}^N, ~ 
\CS_{ij}=\left\{\begin{array}{ll}\displaystyle{(\nabla_{\Bz} \otimes \nabla_{\Bw})\left(\frac{1}{4\pi|\Bz-\Bw|}\right)\Big|_{\substack{\Bz=\BO^{(i)}\\ \Bw=\BO^{(j)}}}} & \quad \text{ if } i \ne j\\
& \\
0I_3 
&\quad \text{otherwise} , \end{array}  \right.  \label{alg2} \\
 \quad \mathbf{Q}&&=\mbox{diag}\{\BCQ^{(1)}, \dots, \BCQ^{(N)}\}\; ~~\mbox{is negative definite}.  \label{alg3}
\end{eqnarray}
These are $3N \times 3N$ matrices whose entries are $3\times 3$ blocks. The notation in \eq{alg2} 
is interpreted as 
$$\CS_{ij} = \left\{\fr{1}{4 \pi}\fr{\prt}{\prt z_q}\Big(  \fr{z_r - O_r^{(j)}}{| \Bz - \BO^{(j)} |^3}\Big)\Big|_{\Bz=\BO^{(i)}}\right\}^3_{q, r=1}
~\mbox{ when  } i\neq j.$$

We 
use the piecewise constant vector function
\beq
\BGX(\Bx) = \left\{  \begin{array}{cc}
\BCQ^{(j)} \BC^{(j)}, ~~\mbox{when} ~~ \Bx \in \ov{B}_{d/4}^{(j)}, ~ j = 1, \ldots, N, \\
\\
0, ~~~\mbox{otherwise}, \end{array} \right.
\eequ{alg3a}
where $B_r^{(j)} = \{\Bx: |\Bx - \BO^{(j)}| < r \}.$

\begin{thm}   \label{thm1_alg} 
Assume that $
 \Gl_{max} < \emph{\text{const }} d^3$, where $\Gl_{max}$ is the largest eigenvalue of the positive definite matrix $-\BQ$ and the constant is independent of $d$. Then the algebraic system 
{\rm \eq{alg1}} is solvable and the vector coefficients $\BC^{(j)}$ satisfy the estimate 
\beq
\sum_{j=1}^N | (\BC^{(j)})^T \BCQ^{(j)} \BC^{(j)} | \leq (1- \emph{\text{const }}\fr{
\Gl_{max} }{ d^3}  )^{-2} \sum_{j=1}^N |(\Grad v
(\BO^{(j)}))^T \BCQ^{(j)} \Grad v
(\BO^{(j)})|.
\eequ{alg4}
\end{thm}


We 
consider the scalar product of \eq{alg1} and the vector $\BQ \BC$:
\beq
\langle \BC, \BQ \BC \rangle + 
\langle \BCS \BQ \BC, \BQ \BC \rangle = -\langle \BGT,
\BQ \BC \rangle.
\eequ{alg6}
Prior to the proof of 
Theorem we formulate and prove the following identity.

\begin{lem} \label{lem1_alg}
 a)
 The scalar product $\langle   \BCS \BQ \BC, \BQ \BC  \rangle $ admits the representation  
\begin{eqnarray}
\langle   \BCS \BQ \BC, \BQ \BC  \rangle &&= \fr{  576
}{\pi^3 
d^6} \int_{{\Bbb R}^3}  \int_{{\Bbb R}^3} \fr{1}{| \BX - \BY  |} (\Grad \cdot \BGX(\BX))  (\Grad \cdot \BGX(\BY)) d \BY  d \BX \nonumber \\
&& - \fr{16}{ \pi d^3}
\sum_{j=1}^N   | \BCQ^{(j)}  \BC^{(j)} |^2.
\label{alg5m}
\end{eqnarray}

b) The following estimate holds
\[|\langle \BCS \BQ \BC, \BQ \BC\rangle| \le \text{\emph{const }}d^{-3}\, \sum_{1 \le j  \le N}|\BQ^{(j)} \BC^{(j)}|^2\;,\]
where the constant in the right-hand side does not depend on $d$.
\end{lem}

{\bf Remark.} Using the notation $\BCN (\Grad \cdot \BGX)$ for the Newton's potential acting on $\Grad \cdot \BGX$ we can interpret the integral in \eq{alg5m} as
$$
\Big(  \BCN (\Grad \cdot \BGX), \Grad \cdot \BGX  \Big)_{L_2({\Bbb R}^3)},
$$
since obviously $\Grad \cdot \BGX \in W^{-1,2}({\Bbb R}^3)$ and $\BCN (\Grad \cdot \BGX) \in W^{1,2}({\Bbb R}^3)$.  Here and in the sequel 
we use the notation
$(  \Gvf, \psi )$ for the extension of the integral $\int_{{\Bbb R}^3} \Gvf (\BX) \psi(\BX) d \BX$ onto the Cartesian product $W^{1,2}({\Bbb R}^3) \times W^{-1,2}({\Bbb R}^3).$

{\bf Proof of Lemma \ref{lem1_alg}.}  
a) By \eq{alg2}, \eq{alg3}, the following representation holds
\beq
 \langle \BCS \BQ \BC, \BQ \BC \rangle 
 = \fr{1}{4 \pi} \sum_{j=1}^N   \Big(   \BCQ^{(j)} \BC^{(j)} \Big)^T  \sum_{1 \leq k \leq N, k \neq j} 
  (\nabla_{\Bz} \otimes \nabla_{\Bw})\left(\frac{1}{ |\Bz-\Bw|}\right)\Big|_{\substack{\Bz=\BO^{(j)}\\ \Bw=\BO^{(k)}}}  \Big(   \BCQ^{(k)} \BC^{(k)} \Big) .
 \eequ{alg6a}
Using the mean value theorem for harmonic functions we note that when $j \neq k$
$$
(\nabla_{\Bz} \otimes \nabla_{\Bw})\left(\frac{1}{ |\Bz-\Bw|}\right)\Big|_{\substack{\Bz=\BO^{(j)}\\ \Bw=\BO^{(k)}}} = \fr{3}{4 \pi (d/4)^3} \int_{B_{d/4}^{(k)}} 
(\nabla_{\Bz} \otimes \nabla_{\Bw})\left(\frac{1}{ |\Bz-\Bw|}\right)\Big|_{\substack{\Bz=\BO^{(j)}}} d \Bw.
$$
Substituting this identity into \eq{alg6a} and using 
 definition \eq{alg3a} we see that the inner sum on the right-hand side of \eq{alg6a} can be presented in the form
\begin{eqnarray}
&& 
\fr{48}{\pi d^3}  \lim_{\tau \to 0+} \int_{{\Bbb R}^3 \setminus B^{(j)}_{(d/4) - \tau} } \Big\{    \fr{\prt}{\prt Y_q}  \Big(   \fr{Y_r - O_r^{(j)}}{|\BY - \BO^{(j)}|^3}\Big)\Big\}_{q,r=1}^3 
\BGX(\BY) d\BY,
\nonumber 
\end{eqnarray}
and further integration by parts gives
\begin{eqnarray}
\langle  \BCS \BQ \BC, \BQ \BC  \rangle &&= - \fr{12}{\pi^2 d^3} \sum_{j=1}^N \Big(   \BCQ^{(j)} \BC^{(j)} \Big)^T
\label{alg7} \\
&&
\cdot \lim_{\tau \to 0+} \Bigg\{  
\int_{{\Bbb R}^3 \setminus B^{(j)}_{(d/4) - \tau} } \Big\{    \fr{Y_r - O_r^{(j)}}{|\BY - \BO^{(j)}|^3}  \Grad \cdot  \BGX(\BY)\Big\}_{r=1}^3 d\BY
\nonumber \\
&&
+ \int_{|\BY - \BO^{(j)}| = (d/4) -\tau} \Bigg\{  \fr{(Y_r - O_r^{(j)})  (Y_q- O_q^{(j)})}{|\BY - \BO^{(j)}|^4}   \Bigg\}_{r, q = 1}^3 d S_\BY ~\BCQ^{(j)} \BC^{(j)} \Bigg\}, \nonumber
\end{eqnarray}
where the integral over ${\Bbb R}^3 \setminus B_{(d/4) - \tau}$ in \eq{alg7} is understood in the sense of distributions. The surface integral  in \eq{alg7} can be evaluated explicitly, i.e.
\beq
\int_{|\BY - \BO^{(j)}| = (d/4) -\tau} \Bigg\{  \fr{(Y_r - O_r^{(j)})  (Y_q- O_q^{(j)})}{|\BY - \BO^{(j)}|^4}   \Bigg\}_{r, q = 1}^3 d S_\BY ~\BCQ^{(j)} \BC^{(j)} = \fr{4 \pi}{3} \BCQ^{(j)} \BC^{(j)}.
\eequ{alg8}
Once again, applying the mean value theorem for harmonic functions in the outer sum of  \eq{alg7} and using \eq{alg8} together with the definition \eq{alg3a} we arrive at
\begin{eqnarray}
\langle  \BCS \BQ \BC, \BQ \BC  \rangle &&= - \fr{16}{\pi d^3} \sum_{j=1}^N |   \BCQ^{(j)} \BC^{(j)} |^2
\label{alg9} \\
- \fr{576 
}{\pi^3 
d^6 } && \lim_{\tau \to 0+}  
\sum_{j=1}^N \int_{ 
B^{(j)}_{(d/4) + \tau}} \int_{{\Bbb R}^3 \setminus B^{(j)}_{(d/4) - \tau} }\sum_{r=1}^3 \GX_r(\BX) \fr{\prt}{\prt X_r}  \Big(\fr{1}{|\BY - \BX|}\Big)  \Grad \cdot  \BGX(\BY) d\BY d \BX,
\nonumber 
\end{eqnarray}
where $\Xi_r$ are the components of the vector function $\BGX$ defined in (\ref{alg3a}).
 
The last integral is understood in the sense of distributions. 
Referring to the definition  \eq{alg3a}, integrating by parts, and taking the limit as
$\tau \to 0+$ we deduce that
the integral term in \eq{alg9} can be written as 
\begin{eqnarray}
\fr{576 
}{\pi^3
d^6 } && 
\int_{{\Bbb R}^3 } \int_{{\Bbb R}^3 }  \fr{1}{|\BY - \BX|} \Big( \Grad \cdot \BGX(\BX) \Big)  \Big( \Grad \cdot  \BGX(\BY) \Big) d\BY d \BX \label{alg10}
\end{eqnarray}
Using \eq{alg9} and \eq{alg10} we arrive at \eq{alg5m}. 

\vspace{.2in}

b) Let us introduce a piece-wise constant function 
\[\BCC(\Bx)=\left\{\begin{array}{ll}
\BC^{(j)}\;, & \quad \text{when } \Bx \in \overline{B^{(j)}_{d/4}}\;, \quad j=1, \dots, N\;,\\
0\;,&\quad \text{otherwise}\;.
\end{array}\right.
\]
According to the system (\ref{R3formalgeq9}), $\nabla \times \BCC(\Bx)=\BO$, and one can use the representation
\begin{equation}\label{CeqgrW}
\BCC(\Bx)=\nabla W(\Bx)
\end{equation}
where $W$ is a scalar function with compact support, and (\ref{CeqgrW}) is understood in the sense of distributions. We give a proof for the case when all voids are spherical, of diameter $\varepsilon$, and hence $\BCQ^{(j)}=-\frac{\pi}{4}\varepsilon^3 I_3$, where $I_3$ is the identity matrix. Then according to (\ref{alg9}) we have 
\begin{eqnarray*}
|\langle  \BCS \BQ \BC, \BQ \BC\rangle| &\le& \frac{16}{\pi d^3} \sum_{1\le j\le N} |\BQ^{(j)} \BC^{(j)}|^2+\frac{36\varepsilon^6}{\pi d^3}\Big| \int_{\mathbb{R}^3}\int_{\mathbb{R}^3} \Big(\nabla_{\BX} W(\BX) \cdot \nabla_{\BX}\Big(\frac{1}{|\BY-\BX|}\Big)\Big)\Delta_{\BY} W(\BY)\, d\BY d\BX \Big| \\
&\le&\frac{16}{\pi d^3} \sum_{1\le j\le N} |\BQ^{(j)} \BC^{(j)}|^2+\frac{144\varepsilon^6}{d^3} \sum_{1 \le j\le N} \int_{B^{(j)}_{d/4}} |\nabla W(\BY)|^2\, d\BY\\
&\le &\frac{\text{const}}{ d^3} \sum_{1 \le j \le N}|\BCQ^{(j)} \BC^{(j)}|^2\;.
\end{eqnarray*}\hfill$\Box$

{\bf Proof of Theorem \ref{thm1_alg}.}  Consider the equation \eq{alg6}. The 
absolute value of its right-hand side  does not exceed
$$
\langle  \BC, -\BQ \BC \rangle^{1/2}  \langle \BGT,
-\BQ \BGT
 \rangle^{1/2}.
$$
Using Lemma 1 and part b) of Lemma 2 we derive
\begin{eqnarray}
\langle \BC, -\BQ \BC \rangle - \text{const} d^{-3}  \langle -\BQ \BC, -\BQ \BC \rangle \leq \langle  \BC, -\BQ \BC \rangle^{1/2}  \langle \BGT,
-\BQ \BGT
\rangle^{1/2}, \nonumber
\end{eqnarray}
leading to
\begin{eqnarray}
\Big(1 -  \fr{\text{const}  
}{ d^3}  \fr{\langle -\BQ \BC, -\BQ \BC \rangle}{\langle \BC, -\BQ \BC \rangle}  \Big) \langle \BC, -\BQ \BC \rangle^{1/2}  \leq  \langle \BGT,
-\BQ \BGT
\rangle^{1/2},
\nonumber
\end{eqnarray}
which implies
\begin{eqnarray}
\Big(1 -  \text{const} 
  \frac{\Gl_{max}}{d^{3}} \Big)^2 \langle \BC, -\BQ \BC \rangle  \leq  \langle \BGT,
-\BQ \BGT
\rangle .
\label{alg11}
\end{eqnarray}
The proof is complete. $\Box$ 

Assuming that the eigenvalues of the matrices $-\BCQ^{(j)}$
are strictly positive 
and satisfy the inequality \eq{Gl_min}, we also find that Theorem \ref{thm1_alg} 
yields 

\begin{coro} \label{coro1_alg} 
Assume that 
the inequalities {\rm \eq{Gl_min}} hold for $\Gl_{max}$ and $\Gl_{min}$.
 Then the vector coefficients $\BC^{(j)}$ in the system {\rm \eq{alg1}} satisfy the estimate 
\beq
\sum_{1 \leq j \leq N} |\BC^{(j)}|^2 \leq \mbox{\rm const  } d^{-3} \| \Grad v
\|^2_{L^2(\Go)},  
\eequ{alg12}
where the constant 
depends only on the coefficients $A_1$ and $A_2$ in {\rm \eq{Gl_min}}.
\end{coro}

\begin{proof}
According to the inequality \eq{alg4} of Theorem \ref{thm1_alg} we deduce
\beq
\Gl_{min} \sum_{1 \leq j \leq N} |\BC^{(j)}|^2 \leq (1- \fr{\text{const} 
}{ d^3} \Gl_{max} )^{-2} \Gl_{max} \sum_{1 \leq j \leq N} |\Grad v
(\BO^{(j)})|^2.
\eequ{alg12aa}
We note  that $v$ 
is harmonic in a neighbourhood of $\ov{\Go}$. Applying the mean value theorem
for harmonic functions together with the Cauchy inequality we write
$$
|\Grad v
(\BO^{(j)})|^2 \leq \fr{48}{ \pi d^{3}} \|\Grad v
\|^2_{L_2(B^{(j)}_{d/4})}.
$$ 
Hence, it 
follows from \eq{alg12aa} that
\begin{eqnarray}
 \sum_{1 \leq j \leq N} |\BC^{(j)}|^2 && \leq d^{-3} (1- \fr{\text{const} 
 }{ d^3} \Gl_{max} )^{-2}  \fr{48}{ \pi}  \fr{\Gl_{max}}{\Gl_{min}}    \sum_{1 \leq j \leq N} \|\Grad v
 \|^2_{L_2(B^{(j)}_{d/4})} \nonumber \\
 && \leq d^{-3 }\left((1- \fr{\text{const} 
 }{ d^3} \Gl_{max} )^{-2} \fr{48}{\pi} \fr{\Gl_{max}}{ \Gl_{min}}  \right) \|\Grad v
 \|^2_{L_2(\Go)},
\label{alg12b}
\end{eqnarray}
which 
is the required estimate \eq{alg12}.
\end{proof}

\section{
Energy error estimate in the case $\GO= {\Bbb R}^3$}\label{R3energyesy}
In this section we 
prove the result concerning the asymptotic approximation of $u_N$ for the 
perforated domain $\Omega_N=\mathbb{R}^3 \backslash \overline{\cup_{j=1}^N F
^{(j)}}$. The changes in the argument, necessary for the treatment of a general domain, will be described in
Section \ref{boundformalg}.  

{\em Proof of Theorem {\rm \ref{thm1_alg_f_inf}}.}
\emph{a) 
Neumann problem for the remainder.} The remainder term  $\CR_N$ in 
\eq{introeq1} is a harmonic function in $\GO_N$, which vanishes at infinity and satisfies the 
boundary conditions 
\beq
\fr{\prt \CR_N}{\prt n}(\Bx) = - \Big(\Grad v(\Bx) + \BC^{(j)} \Big) \cdot \boldsymbol{n}^{(j)}(\Bx) - \sum_{\substack{ k \neq j \\ 1 \leq k \leq N}} \BC^{(k)} \cdot \fr{\prt}{\prt n} \BCD^{(k)}(\Bx), 
~~\mbox{when}~~ \Bx \in \prt 
F^{(j)}, j = 1, \ldots, N.
\eequ{rem1} 
Since $\mbox{supp } f $ 
is separated from $
F^{(j)}, j = 1, \ldots, N,$ and since $\BCD^{(j)}, ~j =1,\ldots,N,$ satisfy  \eq{modfieldeq1} we have 
\beq
\int_{\prt 
F^{(j)}}  \fr{\prt \CR_N}{\prt n}(\Bx) dS_\Bx = 0, ~ j=1,\ldots, N.
\eequ{rem2}
\emph{b) Auxiliary functions.} Throughout the proof we use the notation $B^{(k)}_\rho=\{\X:|\X-\Oj^{(k)}|<\rho\}$.
We 
introduce 
auxiliary functions which will help us to obtain 
\eq{introeq2}. Let 
\begin{eqnarray} \nonumber
\Psi_k(\X)&=&v
(\X)
-v
(\BO^{(k)})-(\X-\Oj^{(k)})\cdot \nabla v(\Oj^{(k)})+ 
\sum_{\substack{ 1 \le j \le N \\ j \ne k  }} \boldsymbol{C}^{(j)} \cdot \BCD^{(j)}(\Bx)  
\\
&&- 
\sum_{\substack{    1 \le j \le N \\   j \ne k }} (\X-\Oj^{(j)})\cdot T(\Oj^{(k)}, \Oj^{(j)}) \BCQ^{(j)}\boldsymbol{C}^{(j)}  
 \;, 
\label{R3energyest3}
\end{eqnarray}
for all $\Bx \in \GO_N$ and $k=1,\dots,N$. 
Every function $\Psi_k$ satisfies
\begin{equation}\label{R3energyest4}
-\Delta \Psi_k(\X)=f(\X)\;, \quad \X \in \Omega_N\;,
\end{equation}
and since $\omega    \cap    \text{supp } f=\varnothing$, we 
see that $\Psi_k, ~k =1, \ldots, N,$ are harmonic in $\Go$. 
Since the coefficients $\BC^{(j)}$ satisfy 
system \eq{alg1}, 
we obtain 
\begin{equation}\label{R3energyest5}
\Dn{\Psi_k}{}(\X)+ \Dn{\mathcal{R}_N}{}(\X)=0\;, \quad \X \in \partial 
F^{(k)}\;.
\end{equation}
and according to \eq{rem2} the functions $\Psi_k$ have 
zero flux through the boundaries of small voids $
F^{(k)}$, i.e. 
\begin{equation}\label{R3energyest6}
\int_{\partial 
F^{(k)}} \Dn{\Psi_k}{}(\X)\, d\X=0\;, k=1, \dots, N\;.
\end{equation}

Next, we introduce smooth cutoff functions $$\chi_\Gve^{(k)}: \Bx \to \chi
((\Bx - \BO^{(k)})/ 
\Gve 
), ~ k= 1, \ldots, N,$$ 
equal to $1$ on $B_{2 \Gve}^{(k)}$ and vanishing outside $B_{3 \Gve}^{(k)}$. 
Then 
by \eq{R3energyest5} we have  
\beq
\fr{\prt}{\prt n}\Big(\CR_N(\Bx) + \sum_{1 \leq k \leq N} \chi^{(k)}_\Gve(\Bx) 
\Psi_k(\Bx)\Big) = 0 ~~\mbox{on}~~ \prt 
F^{(j)}, ~ j = 1, \ldots, N.
\eequ{rem_1a}

 \emph{c) 
 Estimate of the energy integral of  $\mathcal{R}_N$  in terms of  $\Psi_k$.}
 Integrating by parts in $\GO_N$ and using the definition of 
 $\chi^{(k)}_\Gve 
 $, 
  we write   the identity
  \begin{eqnarray} 
   \int_{\GO_N 
   }\nabla \mathcal{R}_N
   \cdot \nabla \Big( \mathcal{R}_N
    + \sum_{1 \leq k \leq N} \chi^{(k)}_\Gve 
   \Psi_k
   \Big)\, d\X 
 = -  \int_{\GO_N 
 }\mathcal{R}_N
  \GD \Big( \mathcal{R}_N
  + \sum_{1 \leq k \leq N} \chi^{(k)}_\Gve 
 \Psi_k
 \Big)\, d\X,
  \label{rem_1B}
 \end{eqnarray}
which is equivalent to 
 \begin{eqnarray}
   \int_{\GO_N 
   }\big| \nabla \mathcal{R}_N\big|^2 d\X + \sum_{1 \leq k \leq N}  \int_{B_{3 \Gve}^{(k)} \setminus \ov{F}^{(k)}
   }
  \nabla \mathcal{R}_N \cdot \nabla \big(  \chi^{(k)}_\Gve  
   \Psi_k
   \big)\, d\X = -  \sum_{1 \leq k \leq N} \int_{ B_{3 \Gve}^{(k)} \setminus \ov{F}^{(k)}  
 }\mathcal{R}_N
  \GD \big(  \chi^{(k)}_\Gve  \Psi_k
 \big)\, d\X,
  \label{rem_1BB}
 \end{eqnarray}
since $\mathcal{R}_N$ is harmonic in $\Omega_N$.

We preserve the notation $\CR_N$ for an extension of $\CR_N$ onto the union of voids $F^{(k)}$ with preservation of the class $W^{1,2}$. Such an extension can be constructed
by using only values of $\CR_N$ on the sets $B_{
2 \Gve}^{(k)} \setminus \ov{F}^{(k)}$ in such a way that
\beq  
\| \Grad 
{\CR}_N \|_{L^2(B_{
2 \Gve}^{(k)})} \leq \mbox{\rm const} \| \Grad \CR_N \|_{L^2(B_{
2 \Gve}^{(k)} \setminus \ov{F}^{(k)})}.
\eequ{CN}
The above fact follows by dilation $\Bx \to \Bx / \Gve$ from the well-known extension theorem for domains with Lipschitz boundaries (see Section 3 of Chapter 6 in \cite{stein}). 
We 
shall use the notation $\ov{\CR}^{(k)}$ for the mean value of $
{\CR}_N$ on $B_{
3 \Gve}^{(k)}$.

The 
integral  on the right-hand side of \eq{rem_1BB} can be written as
\begin{eqnarray}
-\sum_{1 \leq k \leq N}   \int_{B_{
3 \Gve 
}^{(k)} \setminus \ov{F 
}
^{(k)}} \CR_N
\GD \big( \chi^{(k)}_\Gve 
\Psi_k
\big)\, d\X
= - \sum_{1 \leq k \leq N}    \int_{B_{
3 \Gve 
}^{(k)} \setminus \ov{F 
}
^{(k)}} (\CR_N
- 
\ov{\CR}
^{(k)}
)  \GD \big( \chi^{(k)}_\Gve 
\Psi_k
\big)\, d\X, \label{rem_2}
\end{eqnarray}
In the derivation of 
\eq{rem_2} we have used that
\begin{eqnarray}
 &&\int_{B_{
 3 \Gve 
 }^{(k)} \setminus \ov{F 
 }
 ^{(k)}}  \GD \Big( \chi^{(k)}_\Gve 
 \Psi_k 
  \Big)\, d\X 
 = \int_{\prt F 
 ^{(k)}} \fr{\prt   \Psi_k }{\prt n} 
  d S_\X = 0 \label{rem_3}
\end{eqnarray}
according to  \eq{R3energyest6} and  the definition of $\chi^{(k)}_\Gve 
$.

Owing to \eq{rem_1B
} and \eq{rem_2}, we can write 
\begin{eqnarray}
\| \Grad \CR_N \|^2_{L^2(\GO_N)} &&\leq  \GS_1 + \GS_2 , \label{rem_4}
\end{eqnarray} 
where
\beq
\GS_1 = \sum_{1 \leq k \leq N} \Big|  \int_{B_{
3 \Gve 
}^{(k)} \setminus \ov{F
}
^{(k)}}  \Grad \CR_N
\cdot \Grad  \big( \chi^{(k)}_\Gve 
\Psi_k
\big) d\X  \Big| ,
\eequ{rem_44}
and
\beq
\GS_2= \sum_{1 \leq k \leq N} \Big|  \int_{B_{
3 \Gve 
}^{(k)} \setminus \ov{F 
}
^{(k)}}  (\CR_N
- 
\ov{\CR}
^{(k)} )
  \GD \big( \chi^{(k)}_\Gve 
( \Psi_k
- \ov{\Psi}_k  ) \big)\, d\X \Big| ,
\eequ{rem_45}
where 
$\ov{\Psi}_k$ is 
the mean value of $\Psi_k$ over the ball $B_{3 \Gve}^{(k)}$.
Here, we have taken into account that by harmonicity of $\CR_N$,  \eq{rem2} and definition of $\chi_\Gve^{(k)}$
$$
\int_{B_{3 \Gve}^{(k)}\setminus \ov{F}^{(k)}} \GD \Big( \CR_N - \ov{\CR}^{(k)}\Big)
\chi^{(k)}_\Gve 
d \Bx=\int_{B_{3 \Gve}^{(k)}} \GD 
\Big( \CR_N - \ov{\CR}^{(k)}\Big) \chi^{(k)}_\Gve 
d \Bx = 0 .
$$

By the Cauchy inequality, the first sum in \eq{rem_4} allows for the estimate
\begin{eqnarray}
\GS_1 
\leq \Big( \sum_{1 \leq k \leq N}  \|   \Grad \CR_N   \|^2_{L^2(B_{
3 \Gve 
}^{(k)} \setminus \ov{F 
}
^{(k)})}    \Big)^{1/2} \Big( \sum_{1 \leq k \leq N}  \Big\|   \Grad  \big( \chi^{(k)}_\Gve 
\Psi_k
\big)    \Big\|^2_{L^2(B_{
3 \Gve 
}^{(k)} \setminus \ov{F 
}
^{(k)})}    \Big)^{1/2} \;.
\l{rem_55}
\end{eqnarray}
Furthermore, using the inequality 
\beq
 \sum_{1 \leq k \leq N}  \|   \Grad \CR_N   \|^2_{L^2(B_{3 
 \Gve 
 }^{(k)} \setminus \ov{F 
}
^{(k)})}
\leq \|   \Grad \CR_N   \|_{L^2(\GO_N)}^2,
\eequ{rem_56}
together with \eq{rem_55}, we deduce
\begin{eqnarray}
\GS_1 
\leq \|   \Grad \CR_N   \|_{L^2(\GO_N)} \Big( \sum_{1 \leq k \leq N} \big\|   \Grad  \big( \chi^{(k)}_\Gve 
\Psi_k 
\big)    \big\|^2_{L^2(B_{
3 \Gve 
}^{(k)} \setminus \ov{F 
}
^{(k)})}    \Big)^{1/2}. \label{rem_5}
\end{eqnarray} 

%

Similarly to \eq{rem_55}, the second sum in \eq{rem_4} can be estimated as  
\begin{equation}
\GS_2  \leq  \sum_{1 \leq k \leq N}  \Big( \int_{B_{
3 \Gve 
}^{(k)} 
}  (
{\CR}_N 
-   \ov{\CR}
^{(k)}
)^2 d \X \Big)^{1/2}
 \Big( \int_{B_{
 3 \Gve 
 }^{(k)} \setminus \ov{F 
 }
 ^{(k)}}     \big( \GD ( \chi^{(k)}_\Gve 
 ( \Psi_k
 - \ov{\Psi}_k ) ) \big)^2\, d\X   \Big)^{1/2} .  
\label{rem_66}
 \end{equation}
By the Poincar\'{e} inequality for the ball $B^{(k)}_{
 3 \Gve}$ 
 \beq
\|\CR_N - \ov{\CR}^{(k)} \|_{L^2(B^{(k)}_{3 \Gve})}^2 \leq \mbox{\rm const  } \Gve^2 \| \Grad \CR_N \|_{L^2(B^{(k)}_{3 \Gve})}^2
\eequ{Psi_k_ball} 
 we obtain
 $$
 \GS_2 \leq \mbox{const}~ 
 \Gve  ~\Big( \sum_{1 \leq k \leq N}  \|   \Grad 
 {\CR}_N   \|^2_{L^2(B_{
 3 \Gve 
 }^{(k)}) 
}    \Big)^{1/2} \Big( \sum_{1 \leq k \leq N}  \int_{B_{
3 \Gve 
}^{(k)} \setminus \ov{F  
 }
 ^{(k)}}     \big( \GD ( \chi^{(k)}_\Gve 
 ( \Psi_k 
 - \ov{\Psi}_k ) ) \big)^2\, d\X \Big)^{1/2},
 $$ 
  which does not exceed
 \begin{eqnarray}
 &
 & \mbox{const}~ \Gve 
 ~\|   \Grad \CR_N   \|_{L^2(\GO_N)}    \Big( \sum_{1 \leq k \leq N}  \int_{B_{
 3 \Gve 
 }^{(k)} \setminus \ov{F 
 }
 ^{(k)}}     \big( \GD \big( \chi^{(k)}_\Gve 
 ( \Psi_k
 - \ov{\Psi}_k ) \big) \big)^2\, d\X \Big)^{1/2}  , \label{rem_6}
\end{eqnarray} 
because of \eq{CN}. 
Combining \eq{rem_4}--\eq{rem_6} and dividing both sides of \eq{rem_4} by $\| \Grad \CR_N \|_{L^2(\GO_N)}$  we arrive at
\begin{eqnarray}
 \| \Grad \CR_N \|_{L^2(\GO_N)} &&\leq \Big( \sum_{1 \leq k \leq N}  \Big\|   \Grad  \big( \chi^{(k)}_\Gve 
( \Psi_k
- \ov{\Psi}_k  ) \big)   \Big\|^2_{L^2(B_{
 3 \Gve 
 }^{(k)} 
 )}    \Big)^{1/2} \nonumber \\
&& + \mbox{const} ~ \Gve 
~  \Big( \sum_{1 \leq k \leq N}  \int_{B_{
3 \Gve 
}^{(k)} 
}     \big\{ (\Psi_k
-\ov{\Psi}_k) \GD 
\chi^{(k)}_\Gve 
+ 2 \Grad 
\chi^{(k)}_\Gve 
\cdot \Grad \Psi_k
\big\}^2\, d\X \Big)^{1/2}, \label{rem_7}
\end{eqnarray}
which leads to
\begin{eqnarray}
&& \| \Grad \CR_N \|^2_{L^2(\GO_N)} \leq \mbox{const} \sum_{1 \leq k \leq N} \Big( \|   \Grad \Psi_k \|^2_{L^2(B_{
3 \Gve 
}^{(k)} 
)
}  + \Gve 
^{-2} \|   \Psi_k - \ov{\Psi}_k\|^2_{L^2(B_{
3 \Gve 
}^{(k)} 
)
}    \Big). \label{Psi_k_ball_1}
\end{eqnarray}
Applying the Poincar\'{e} inequality (see \eq{Psi_k_ball}) for $\Psi_k$ in the ball $B_{3 \Gve}^{(k)}$ and using \eq{Psi_k_ball_1},  we deduce
%
\begin{eqnarray}
&& \| \Grad \CR_N \|^2_{L^2(\GO_N)} \leq \mbox{const} \sum_{1 \leq k \leq N} 
\|   \Grad \Psi_k \|^2_{L^2(B_{
3 \Gve 
}^{(k)} 
)
}  . \label{Psi_k_ball_2}
\end{eqnarray}

\emph{d) 
Final energy estimate. }Here we prove the inequality \eq{introeq2}. 
 Using 
 definition \eq{R3energyest3} of 
$\Psi_k$, $k=1, \dots, N$, we can replace  the preceding inequality by
\begin{equation}\label{R3energyest12a}
\| \nabla \mathcal{R}_N\|^2_{L^2(\Omega_N)}\le \text{const } \big\{ \mathcal{K}+\mathcal{L} 
\big\}\;,
\end{equation}
where 
\begin{equation}\label{R3energyest13}
\begin{array}{l}\displaystyle{ \mathcal{K}=\sum_{1\le k\le N} \| \nabla v
(\cdot 
)- \nabla v
(\Oj^{(k)})\|^2_{L^2(B_{
3 \Gve 
}^{(k)} 
)}\;,}\\
\\
\displaystyle{\mathcal{L}= 
\sum_{1 \le k \le N} \Big\| \sum_{\substack{j \ne k \\ 1 \le j \le N}}  \Big[\nabla \Big(\boldsymbol{C}^{(j)} \cdot \BCD^{(j)}(\cdot 
) \Big)- 
T(\Oj^{(k)}, \Oj^{(j)})\BCQ^{(j)}\boldsymbol{C}^{(j)} \Big] \Big\|^2_{L^2(B_{
3 \Gve 
}^{(k)} 
)}.}
\end{array} 
\end{equation}

The estimate for $\CK$ 
is straightforward and 
it follows by Taylor's expansions of $v$ 
in the vicinity of $\BO^{(k)}$,  
\beq
\CK 
\leq \text{const }  \Gve^5 d^{-3}  \max_{\Bx \in \ov{\Go}, 1 \le i, j\le 3} \Big|\frac{\partial^2 v
}{\partial x_i \partial x_j}\Big|^2. 
\eequ{rem_8}
Since $v$ 
is harmonic 
in a neighbourhood of $\ov{\Go}$, we obtain by the 
local regularity property of harmonic functions that
\beq
\CK  
\leq \text{const } \Gve^5 d^{-3}  \big\| \Grad v
\big\|^2_{L^2(  {\Bbb R}^3 ) 
}. 
\eequ{rem_88}

To estimate $\CL$,  
we use Lemma \ref{Djasymp} on the asymptotics of the dipole fields together with the definition \eq{defT} of the matrix function $T$, which  lead to 
 \begin{equation}\label{R3energyest14}
 |\nabla(\boldsymbol{C}^{(j)} \cdot \BCD^{(j)}(\Bx))-  T(\Oj^{(k)}, \Oj^{(j)})\BCQ^{(j)} \boldsymbol{C}^{(j)}| \le \text{const }\varepsilon^4|\boldsymbol{C}^{(j)}| |\X-\Oj^{(j)}|^{-4}\;, 
 \end{equation}
 for $ \X \in B^{(k)}_{
 3 \Gve 
 }.$ 
 Now, it follows from \eq{R3energyest13} and \eq{R3energyest14} 
 that
\begin{eqnarray}
\mathcal{L} &\le& \text{const } \varepsilon^8 \sum^N_{k=1} \int_{B_{3 \Gve}^{(k)}} 
\Big(   \sum_{1 \leq j \leq N, j \neq k} \fr{| \BC^{(j)} |}{|\Bx - \BO^{(j)}|^4} \Big)^2 d \Bx, \label{rem_9}
\end{eqnarray}
and by the Cauchy inequality 
the right-hand side  does not exceed
\begin{eqnarray}
&\mbox{const  }& \Gve^8 \sum_{p=1}^N |\BC^{(p)}|^2 \sum_{k=1}^N \sum_{1 \leq j \leq N, j \neq k} \int_{B_{3 \Gve}^{(k)}}  \fr{d \Bx}{|\Bx - \BO^{(j)}|^8 } \leq \mbox{const  } \Gve^{11}\sum_{p=1}^N |\BC^{(p)}|^2 \sum_{k=1}^N \sum_{1 \leq j \leq N, j \neq k} 
 \fr{1 
 }{|\BO^{(k)} - \BO^{(j)}|^8 } 
 \nonumber \\
 &&\leq  \mbox{const  } \fr{\Gve^{11}}{d^6} \sum_{p=1}^N |\BC^{(p)}|^2 
 \int\int_{\{\omega \times \omega:  |\BX - \BY| > d\}} \fr{d \BX d \BY}{|\BX - \BY|^8} \leq \mbox{const  } \fr{\Gve^{11}}{d^8} \sum_{p=1}^N |\BC^{(p)}|^2 . \label{rem_99}
 \end{eqnarray}

Since the eigenvalues of the matrix $-\BQ$ satisfy the constraint \eq{Gl_min},  we can apply Corollary \ref{coro1_alg} and  use
 the 
estimate \eq{alg12} for the 
right-hand side of \eq{rem_99} 
to obtain 
\beq
\CL 
\leq \mbox{const}~ \Gve^{11} 
d^{-11}
\|\Grad v
\|^2_{L^2(\Go)}.
\eequ{rem_11}
Combining \eq{R3energyest12a}, \eq{rem_88} and \eq{rem_11},  we arrive at \eq{introeq2} 
 and complete the proof. 
$\Box$
\section{
Approximation of $u_N$ for 
a perforated 
domain}\label{boundformalg}
Now we seek 
an approximation of 
the solution $u_N$  to the problem 
\eq{intro_2}--\eq{intro_4} 
assuming that $\Omega$ is 
an arbitrary domain in $\mathbb{R}^3$. We first describe the formal asymptotic algorithm and 
derive 
a system of algebraic equations, similar to \eq{alg1}, 
 which is used 
 for evaluation of the coefficients in the asymptotic representation of $u_N$.
 

\subsection{Formal asymptotic algorithm for the 
perforated domain $\GO_N$}

The solution $u_N \in L^{1,2}(\GO_N)$  of  \eq{intro_2}--\eq{intro_4}   
is sought in the form 
\beq
u_N(\X)=v
(\X)+ \sum_{k=1}^N \boldsymbol{C}^{(k)}\cdot \Big\{\BCD^{(k)}(\Bx) - 
\BCQ^{(k)}\nabla_\Y H(\X, \Y)\big|_{\Y=\Oj^{(k)}} \Big\}+R_N(\X)\;,
\eequ{u_N_0}
where in this instance $v$ 
solves 
problem \eq{vfom_1}, \eq{vfom_2}   in  Section \ref{Modelproblems}, and $R
_N$  is a harmonic function in $\Omega_N$.
%
Here $\boldsymbol{C}^{(k)}$, $k=1, \dots, N$ are the vector coefficients to be determined.

Owing to the definitions of $\BCD^{(k)}$, $k=1, \dots, N,$ and $H$ as solutions of Problems 2 and 3 in Section \ref{Modelproblems}, 
and taking into account Lemma \ref{Djasymp} 
on the asymptotics of $\BCD^{(k)}$ 
 we deduce that $|R_N(\Bx)|$ is small for $\Bx \in \prt \GO.$

On the boundaries 
$\prt F
^{(j)}$, the substitution of (\ref{u_N_0}) 
into (\ref{intro_4}) 
yields
\begin{eqnarray*}
\Dn{{R}_N}{}(\X)&=&-\N^{(j)} \cdot \Big\{ \nabla v
(\Oj^{(j)})+\boldsymbol{C}^{(j)}+O(\varepsilon)+O(\varepsilon^3 |\boldsymbol{C}^{(j)}|)
\\&&+ 
\sum_{\substack{k \ne j\\ 1 \le k \le N}} \nabla \Big\{ \boldsymbol{C}^{(k)} \cdot \Big(\BCD^{(k)}(\Bx 
)  - 
\BCQ^{(k)}\nabla_\Y H(\X, \Y)\big|_{\Y=\Oj^{(k)}} \Big) \Big\}\Big\}
\;, \quad \X \in  \partial F
^{(j)}, j=1, \dots, N\;.
\end{eqnarray*}
Then, using the asymptotic representation (\ref{modfieldeq2}) 
in Lemma \ref{Djasymp} we deduce
\begin{eqnarray} 
\label{boundformalalgeq3}
\Dn{{R}_N}{}(\X)&\sim&-\N^{(j)}\cdot\Big\{ \nabla v
(\Oj^{(j)})+\boldsymbol{C}^{(j)}
+    
\sum_{\substack{k \ne j \\ 1 \le k \le N}} {\frak T}(\X, \Oj^{(k)})\BCQ^{(k)} \boldsymbol{C}^{(k)}  
\Big\}\;, \quad \X \in \partial F 
^{(j)}, j=1, \dots, N\;,
\end{eqnarray}
where ${\frak T(\X, \Y)}$ is defined by %
\begin{equation}\label{defFT}
{\frak T}(\X, \Y)=(\nabla_{\X} \otimes \nabla_{\Y}) G(\X, \Y)  \;,
\end{equation} 
with $G(\X, \Y)$ being Green's function for the 
domain $\GO$, as defined in Section \ref{Modelproblems}.
To compensate for the leading discrepancy 
in the boundary conditions \eq{boundformalalgeq3}, 
we choose  the coefficients $\boldsymbol{C}^{(m)}$, $m=1, \dots, N,$ 
subject to  the algebraic system 
\beq
\nabla v
(\Oj^{(j)})+\boldsymbol{C}^{(j)} +
\sum_{\substack{k \ne j \\ 1 \le k \le N}} {\frak T}(\Oj^{(j)}, \Oj^{(k)})
\BCQ
^{(k)} \boldsymbol{C}^{(k)}=0, ~~ j =1, \ldots,N,
\eequ{R3formalgeq9a}
where $\BCQ
^{(k)}, k=1,\ldots,N,$ are polarization matrices of small voids $
F^{(k)}$, as  in Lemma \ref{Djasymp}. 

Provided 
system \eq{R3formalgeq9a} has been solved for the vector coefficients $\BC^{(k)}$, 
 formula \eq{u_N_0} 
leads to the formal asymptotic approximation of $u_N$:
\beq
u_N(\Bx) \sim v
(\Bx) + \sum_{k=1}^N \boldsymbol{C}^{(k)}\cdot \Big\{\BCD^{(k)}(\Bx) - 
\BCQ^{(k)}\nabla_\Y H(\X, \Y)\big|_{\Y=\Oj^{(k)}} \Big\}.
\eequ{u_N_formal}

\subsection{Algebraic system} 

The system \eq{R3formalgeq9a} can be written in  the matrix form
 \beq
\BC + 
{\frak S} \BQ \BC = - \BGT, 
\eequ{alg11a}
where
\beq
{\frak S}=[{\frak S}_{ij}]_{i, j=1}^N, ~ 
{\frak S}_{ij}=\left\{\begin{array}{ll}\displaystyle{(\nabla_{\Bz} \otimes \nabla_{\Bw})G(\Bz, \Bw )\Big|_{\substack{\Bz=\BO^{(i)}\\ \Bw=\BO^{(j)}}}} & \quad \text{ if } i \ne j\\
& \\
{0}I_3 
&\quad \text{otherwise}\end{array}  \right.  
\eequ{alg22a}
with $G(\Bz, \Bw)$ 
standing for Green's function in the limit domain $\GO$,
and the block-diagonal matrix $\BQ$ being the same as in \eq{matrix_Q}. 
The system \eq{alg11a} is similar to that 
in Section \ref{alg_syst}, with the only change 
of the matrix $\BCS$ 
for ${\frak S}$. 
The elements of  ${\frak S}$ are given via the second-order derivatives of Green's function in $\GO$, as defined in \eq{defFT}. 
The next assertion is similar to Corollary \ref{coro1_alg}.

 \begin{lem}  \label{thm1_alg_f} 
Assume that 
inequalities {\rm \eq{Gl_min}} hold for $\Gl_{max}$ and $\Gl_{min}$. 
Also 
let $v
$ be a unique solution of 
problem {\rm \eq{vfom_1}, \eq{vfom_2}} in the
domain $\GO$.  
Then the vector coefficients $\BC^{(j)}$ in the system {\rm \eq{R3formalgeq9a}} satisfy the estimate 
\beq
\sum_{1 \leq j \leq N} |\BC^{(j)}|^2 \leq \mbox{\rm const} ~ d^{-3} \| \Grad v
 \|^2_{L^2(\GO)}, 
\eequ{alg12a}
where the constant depends on the shape of the voids $F^{(j)}, j= 1,\ldots, N.$
\end{lem}

\begin{proof}
The proof of the theorem is very similar to the one given in Section \ref{alg_syst}. We consider the scalar product of \eq{alg11a} and the vector $\BQ \BC$:
\beq
\langle \BC, \BQ \BC \rangle + \langle {\frak S} \BQ \BC, \BQ \BC \rangle =  -\langle  \BGT
, \BQ \BC \rangle,
\eequ{alg112a}
and similarly to \eq{alg5m} derive

\begin{eqnarray}
\langle   {\frak S} \BQ \BC, \BQ \BC  \rangle &&= 
48^2~ \pi^{-2
}~ d^{-6} \int_{\GO}  \int_{\GO} 
G(\BX, \BY)(\Grad \cdot \BGX(\BX))  (\Grad \cdot \BGX(\BY)) d \BY  d \BX \nonumber \\
&& - 16 \pi^{-1} d^{-3}
\sum_{1 \le j \le N }  | \BCQ^{(j)}  \BC^{(j)} |^2  \nonumber \\
&& + \sum_{1 \leq j \leq N}   \Big(   \BCQ^{(j)} \BC^{(j)} \Big)^T 
  (\nabla_{\Bz} \otimes \nabla_{\Bw})\left(H(\Bz, \Bw)\right)\Big|_{\substack{\Bz=\BO^{(j)}\\ \Bw=\BO^{(j)}}}  \Big(   \BCQ^{(j)} \BC^{(j)} \Big) ,
\label{alg5new}
\end{eqnarray}
where the integral in the right-hand side is positive, and it is understood in the sense of distributions, in the same way as in the proof of Lemma \ref{lem1_alg}, while the magnitude of the last sum in \eq{alg5new} is small compared to the magnitude of the second sum. 

Now, 
the right-hand side in \eq{alg112a} does not exceed
$$
\langle  \BC, -\BQ \BC \rangle^{1/2}  \langle \BGT, 
 -\BQ \BGT
 \rangle^{1/2}.
$$
Following the same pattern as in the proof of Theorem \ref{thm1_alg}, we deduce
\begin{eqnarray}
\langle \BC, -\BQ \BC \rangle - \mbox{const } d^{-3}  \langle -\BQ \BC, -\BQ \BC \rangle \leq \langle  \BC, -\BQ \BC \rangle^{1/2}  \langle \BGT
, -\BQ \BGT
\rangle^{1/2}, \nonumber
\end{eqnarray}
where the constant is independent of $d$. Furthermore, this leads to 
\begin{eqnarray}
\Big(1 -  \mbox{const } d^{-3}  \fr{\langle -\BQ \BC, -\BQ \BC \rangle}{\langle \BC, -\BQ \BC \rangle}  \Big) \langle \BC, -\BQ \BC \rangle^{1/2}  \leq  \langle \BGT
, -\BQ \BGT
 \rangle^{1/2},
\nonumber
\end{eqnarray}
which implies
\begin{eqnarray}
\Big(1 -  \mbox{const } d^{-3} \Gl_{max} \Big)^2 \langle \BC, -\BQ \BC \rangle  \leq  \langle \BGT,
-\BQ \BGT
\rangle,
\label{alg1111}
\end{eqnarray}
where $\Gl_{max}$ is the largest eigenvalue of the positive definite matrix $-\BQ$.
Then using the same estimates \eq{alg12aa} and  \eq{alg12b} as in the proof of  Corollary
\ref{coro1_alg} we arrive at \eq{alg12a}.
\end{proof}

\subsection{
Energy estimate for the remainder}\label{boundenergyestRN}

\begin{thm}\label{mainthm1}
Let the parameters $\varepsilon$ and $d$ satisfy the inequality
\[\varepsilon < c\, d   
\;,\]
where $c$ is a sufficiently small absolute constant. Then the solution $u_N(\X)$ of 
{\rm \eq{intro_2}--\eq{intro_4}}   is represented by the asymptotic formula
\begin{equation}\label{introeq1a}
u_N(\X)=v(\X)+ 
 \sum^N_{k=1} \boldsymbol{C}^{(k)} \cdot \{ 
  \BCD^{(k)}(\Bx) 
 - 
 \BCQ^{(k)} \nabla_\Y H(\X, \Y)\Big|_{\Y=\Oj^{(k)}}\}+R_N(\X)\;,
\end{equation}
where $\boldsymbol{C}^{(k)}=(C^{(k)}_1, C^{(k)}_2, C^{(k)}_3)^T$
solve the linear algebraic system {\rm \eq{R3formalgeq9a}}. 
The remainder $R_N$ in {\rm \eq{introeq1a}} satisfies the energy
estimate
\begin{equation}\label{introeq2a}
 \| \nabla R_N\|^2_{L_2(\Omega_N)} \le \text{\emph{const} } \Big\{ \varepsilon^{11}d^{-11}  + \varepsilon^{5}d^{-3}  \Big\}  \| \nabla v
 \|^2_{L_2(\Omega)}
\;.
 \end{equation}
\end{thm}


{\it Proof. } Essentially, the proof follows the same steps as in 
Theorem 
\ref{thm1_alg_f_inf}. Thus, we give an outline indicating the obvious modifications, which are brought 
by the boundary $\prt \GO.$

\emph{a) Auxiliary functions. } Let us preserve the notations $\chi_\Gve^{(k)}$ for cutoff functions used in the proof of Theorem \ref{thm1_alg_f_inf}. We also
need a new cutoff function $\chi_0$ to isolate $\prt \GO$ from the cloud of holes. Namely, let $(1-\chi_0) \in C_0^\infty(\GO)$
and $\chi_0 = 0$ on a neighbourhood of $\ov{\Go}$. A neighbourhood of $\prt \GO$ containing $\mbox{supp } \chi_0$ will be denoted by $\CV$.
%
%
Instead of the functions $\Psi_k$ defined in \eq{R3energyest3}, we introduce
\begin{eqnarray}\label{R3energyest3_bdd}\nonumber
\Psi^{(\GO)}_k(\X)&=&v
(\X)-v
(\BO^{(k)})-(\X-\Oj^{(k)})\cdot \nabla v
(\Oj^{(k)})+ 
\sum_{\substack{j \ne k\\ 1 \le j \le N}} \boldsymbol{C}^{(j)} \cdot \BCD^{(j)}(  \Bx 
)\\
&&-
\sum_{\substack{j \ne k \\ 1 \le j \le N}} (\X-\Oj^{(j)})\cdot {\frak T}(\Oj^{(k)}, \Oj^{(j)}) \BCQ^{(j)}\boldsymbol{C}^{(j)}-\sum_{j=1}^N \boldsymbol{C}^{(j)} \cdot \BCQ^{(j)} \nabla_\By H(\Bx, \By)\Big|_{\By=\BO^{(j)}}\;, 
\end{eqnarray}
where the matrix ${\frak T}$ is defined in \eq{defFT} via second-order derivatives of Green's function in $\GO$. Owing to \eq{introeq1a} and the algebraic system \eq{R3formalgeq9a} we have
\beq
\fr{\prt}{\prt n} \left( \Psi_k^{(\GO)}(\Bx) + R_N(\Bx)  \right) = 0, ~~ \Bx \in \prt F 
^{(k)}.
\eequ{bcond_psi}

We also
 use the  function
\begin{equation}\label{energyestbound1}
\Psi_0(\X)=  
\sum^N_{j=1} \boldsymbol{C}^{(j)}\cdot\big\{\cD^{(j)}(\Bx ) 
 - 
 \BCQ^{(j)}\frac{(\X-\Oj^{(j)})}{4 \pi |\X-\Oj^{(j)}|^3}\big\}\;,
\end{equation}
which is harmonic in $\Omega_N$.  It follows from \eq{introeq1a} that 
\begin{equation}
R_N(\X)+\Psi_0(\X)=-\sum_{1 \leq j \leq N} \BC^{(j)} \cdot  \BCQ^{(j)} \Big\{ \frac{(\X-\Oj^{(j)})}{4 \pi |\X-\Oj^{(j)}|^3} -  \nabla_\Y H(\X, \BY)\Big|_{\BY=\BO^{(j)}} \Big\}  =0 \;, \quad \X \in \partial  \Omega\;. \label{psi0_bcond}
\end{equation}

\emph{b) The energy estimate for $R_N$. 
} 
We start with the identity
 \begin{eqnarray}
  && \int_{\GO_N}\nabla \Big({R}_N
  + \chi_0
  \Psi_0
  \Big) \cdot \nabla \Big( {R}_N
   + \sum_{1 \leq k \leq N} \chi^{(k)}_\Gve 
  \Psi^{(\GO)}_k
  \Big)\, d\X \nonumber \\
   &&
   = - \int_{\GO_N}\Big({R}_N 
   + \chi_0 
   \Psi_0 
    \Big)  \GD \Big( {R}_N 
     + \sum_{1 \leq k \leq N} \chi^{(k)}_\Gve 
   \Psi^{(\GO)}_k 
   \Big)\, d\X, 
    \label{rem_1_bdd}
  \end{eqnarray}
  which follows from \eq{bcond_psi}, \eq{psi0_bcond}  by Green's formula. 
  According to the definitions of 
  $\chi_0$ and $\chi^{(k)}_\Gve$, we have $\mbox{supp } \chi_0 
   \cap \mbox{supp } \chi^{(k)}_\Gve 
   = \emptyset$ 
   for all $k=1,\ldots,N$.
  Hence the integrals in \eq{rem_1_bdd} involving the products of $\chi_0$ and $\chi^{(k)}_\Gve$ or their derivatives are equal to zero. Thus, using that $\GD R_N = 0 $ on $\GO_N$, we reduce 
  \eq{rem_1_bdd} 
 to the equality 
  \begin{eqnarray}
 && \int_{\GO_N} |\Grad R_N 
 |^2 d \Bx + \sum_{1 \leq k \leq N} \int_{B_{3 \Gve} \setminus \ov{F 
 }^{(k)}} \Grad R_N 
 \cdot \Grad \Big(  \chi^{(k)}_\Gve 
 \Psi^{(\GO)}_k
 \Big) d\Bx +  \int_{\GO_N \cap \CV
  } \Grad R_N
  \cdot \Grad \Big(  \chi_0
  \Psi
  _0
  \Big) d\Bx\label{rem_2_bdd}\\
  &&  = -  \sum_{1 \leq k \leq N} \int_{\GO_N} R_N
  \GD \left( \chi^{(k)}_\Gve 
  \Psi^{(\GO)}_k
   \right) d\Bx, \nonumber
  \end{eqnarray}
  which differs  in the left-hand side from \eq{rem_1BB} only by the integral over $\GO_N \cap \CV$.
  
  Similarly to the part (b) of the proof of Theorem \ref{thm1_alg_f_inf}  
  we deduce
  \begin{eqnarray}
&& \| \Grad \CR_N \|^2_{L^2(\GO_N)} \leq \mbox{const}\Big\{  \|   \Grad \Psi_0 \|^2_{L^2(\GO \cap \CV
)}  + 
\|   \Psi_0 \|^2_{L^2(\GO \cap \CV
)} 
+ \sum_{1 \leq k \leq N} 
\|   \Grad \Psi_k \|^2_{L^2(B_{3 \Gve}^{(k)} 
)} 
\Big\}. \label{rem_3_bdd}
\end{eqnarray}

Similar to the 
steps of part (d) of the proof 
in Theorem \ref{thm1_alg_f_inf},  
the last sum is majorized by 
\beq
\mbox{const } (\Gve^{11} d^{-11} + \Gve^5 d^{-3} ) \| \Grad v
\|_{L^2(\GO)}^2.
\eequ{bound_new}
It remains to estimate 
two terms in \eq{rem_3_bdd} containing $\Psi_0$. Using \eq{R3energyest14}, 
together with 
\eq{alg12a}  we deduce 
\begin{eqnarray}
&&
\|\Psi_0
\|^2_{L^2 (\GO \cap \CV
)} \leq \mbox{const} ~ 
\Gve^8 \sum_{1 \leq j \leq N} \int_{\GO \cap \CV
} \fr{|C^{(j)}|^2 d \Bx}{|\Bx - \BO^{(j)}|^6} \nonumber \\
&& \leq \mbox{const  } 
\Gve^8
\sum_{1 \leq j \leq N} |C^{(j)}|^2 \leq \mbox{const} \fr{\Gve^8}{d^3 
} \|\Grad v
\|^2_{L^2(\GO)},  \label{rem_4_bdd_a}
\end{eqnarray}
and
\begin{eqnarray}
&& \| \Grad \Psi_0
\|^2_{L^2 (\GO \cap \CV
)} \leq \mbox{const} ~  \Gve^8 \sum_{1 \leq j \leq N} \int_{\GO \cap \CV
} \fr{|C^{(j)}|^2 d \Bx}{|\Bx - \BO^{(j)}|^8} \nonumber \\
&& \leq \mbox{const  } 
\Gve^8  
\sum_{1 \leq j \leq N} |C^{(j)}|^2 \leq \mbox{const} \fr{\Gve^8}{d^3} 
\|\Grad v
\|^2_{L^2(\GO)}. \label{rem_4_bdd}
\end{eqnarray}
Combining \eq{rem_3_bdd}--\eq{rem_4_bdd} we complete the proof. $\Box$

\section{Illustrative example} \label{example}
  
Now, the asymptotic approximation derived in the previous section is applied to the case of a relatively simple geometry, where all the terms  in the formula \eq{introeq1a} can be written explicitly.

\subsection{The case of a domain with a cloud of spherical voids}
\label{cloud_example}

Let $\GO_N$ be  a ball of a finite radius $R$, with the centre at the origin, containing $N$ spherical voids $F^{(j)}$ of radii $\rho_j$ with the centres at $\BO^{(j)}, j=1,\ldots, N,$ as shown in Fig. \ref{diagramcloud}.
\begin{figure}[htbp]
\centering
\includegraphics[width=0.495\textwidth]{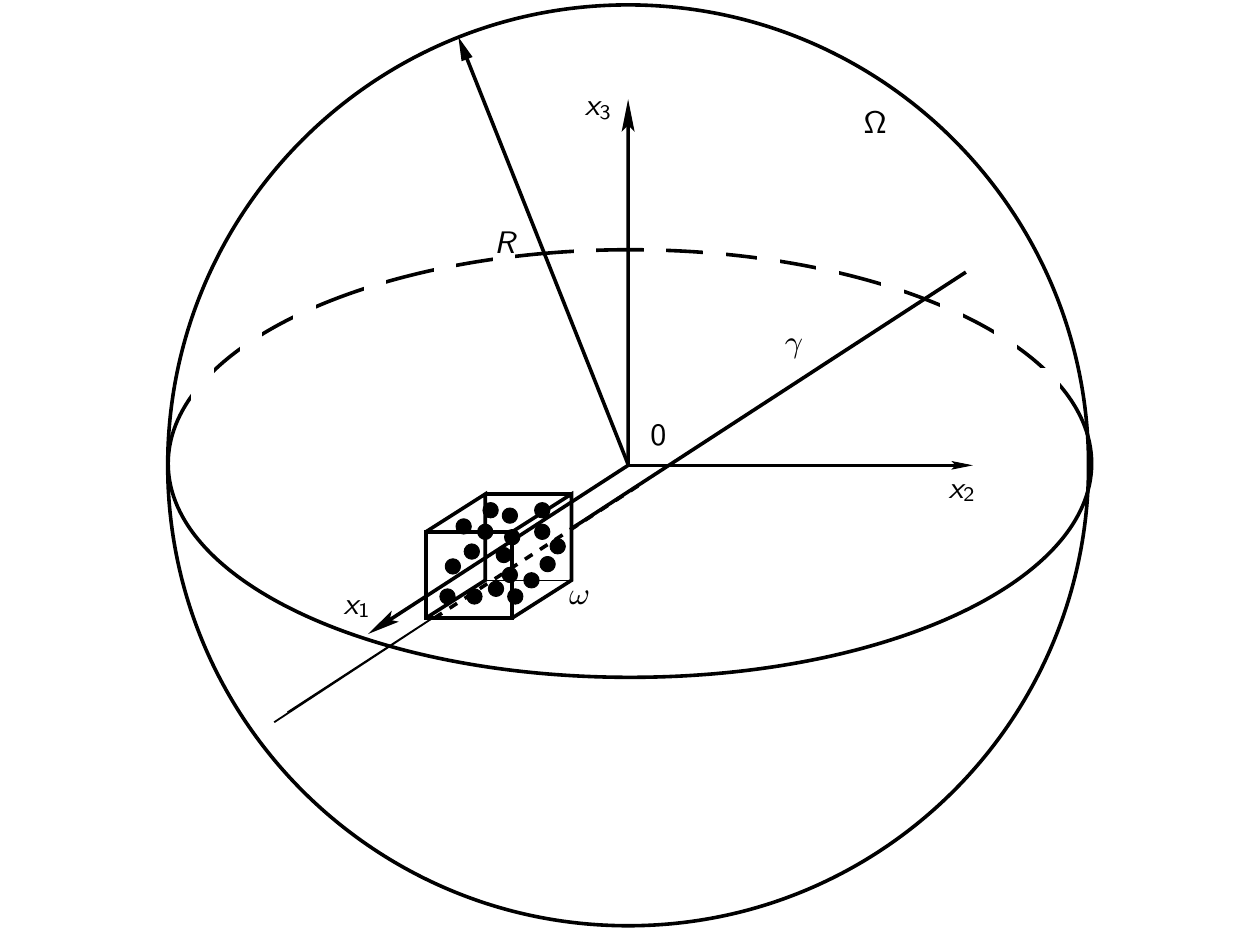}
\caption{Example configuration of a sphere containing a cloud of spherical voids in a the cube $\omega$.}
\label{diagramcloud}
\end{figure}
The radii of the voids are assumed to be smaller than the 
distance between nearest neighbours. 
We put $\phi\equiv 0$ and
\beq
f(\Bx) = \left\{ \begin{array}{cc}
6 
& \mbox{ when }   |\Bx|< \rho, \\
0 & \mbox{ when }    \rho < |\Bx| < R.
\end{array} \right.
\eequ{ex0}
Here, it is assumed that $\rho+ b  < |\BO^{(j)}| < R - b, ~ 1 \leq j \leq N, $ where $\rho$ and $b$ are positive constants independent of $\Gve$ and $d$.

The function $u_N$ is 
the solution of the mixed boundary value problem for the Poisson equation:
\begin{eqnarray}
&&\GD u_N(\Bx) + f(\Bx) = 0, ~~\mbox{when  } \Bx \in \GO_N, \label{ex1} \\
&& u_N(\Bx) = 0,   ~~\mbox{when  } |\Bx| = R, \label{ex2} \\
&& \fr{\prt u_N}{\prt n}(\Bx) = 0, ~~\mbox{when  } |\Bx - \BO^{(j)}| = \rho_j, ~j = 1,\ldots,N. \label{ex3} 
\end{eqnarray}

In this case, $u_N$ is approximated by \eq{introeq1a}, where the solution of the 
Dirichlet problem in $\GO$ is given by
\begin{equation}
v
(\Bx) = \left\{ \begin{array}{cc} 
\rho^2 (3 -2 \rho R^{-1})  - |\Bx|^2 & \mbox{ when }   |\Bx|< \rho, \\
2 \rho^3 (|\Bx|^{-1}- R^{-1}) & \mbox{ when }    \rho < |\Bx| < R.
\end{array} \right.
\eequ{ex4}
In turn, the dipole fields $\BCD^{(j)}$ and the dipole matrices $\BCQ^{(j)}$ have the form
\begin{eqnarray}
\BCD^{(j)}(\Bx) = - \rho_j^3 \fr{\Bx - \BO^{(j)}}{|\Bx - \BO^{(j)}|^3}, ~~\BCQ^{(j)} = - 2 \pi \rho_j^3 I_3, \label{ex5}
\end{eqnarray}
where $I_3$ is the $3\times 3$ identity matrix.

The regular part $H(\Bx, \By)$ of Green's function in the 
domain $\GO$ (see \eq{H_def}) is 
\beq
H(\Bx, \By) = \fr{R}{4 \pi |\By| |\Bx - \hat{\By}|}, ~~ \hat{\By} = \fr{R^2}{|\By|^2} \By. 
\eequ{ex6}
The coefficients $\BC^{(j)}, ~j = 1, \ldots, N,$ in 
\eq{introeq1a} are defined from the algebraic system
\eq{R3formalgeq9a}, where Green's function $G(\Bx, \By)$ is given by
\beq
G(\Bx, \By) = \fr{1}{4 \pi |\Bx - \By|} - \fr{R}{4 \pi |\By| |\Bx - \hat{\By}|}.
\eequ{ex7}

\subsection{Finite elements simulation versus the asymptotic approximation}
The explicit representations of the fields $v, {\BCD}^{(j)}, H, G$, given above, are used in the asymptotic formula
\eq{introeq1a}. Here, we present a comparison between the results of an independent Finite Element computation, produced in COMSOL, and the mesoscale asymptotic approximation \eq{introeq1a}.

For the computational example, we set $R=120$, and consider a cloud of $N=18$ spherical voids arranged into a cloud of a
parallelipiped shape. The  position of the centre and radius of each void is included in Table \ref{Table1a}. The support of the function $f$ (see (\ref{ex0})), is chosen to be inside the sphere with radius $\rho=30$ and  centre at the origin,
as stated in \eq{ex0}.  
\begin{table}[htbp]\centering
\label{tabincdata}
\begin{tabular}{|c|c|c||c|c|c|}
 \hline\hline Void & Centre & $\rho_j /R$ &  Void & Centre & $\rho_j/R$\\
\hline\hline $F^{(1)}$   &(-50, 0, 0)&  0.0417 &  $F^{(10)}$   &(-72, 0, 0)& 0.0417 \\
\hline $F^{(2)}$   &(-50, 0, 22)&  0.0333 &  $F^{(11)}$   &(-72, 0, 22) & 0.0458\\
\hline $F^{(3)}$   &(-50, 22, 0)&  0.0292&  $F^{(12)}$   &(-72, 22, 0) & 0.0292\\
\hline $F^{(4)}$   &(-50, 0, -22)& 0.0375  &  $F^{(13)}$   &(-72, 0, -22) &  0.0375    \\
\hline $F^{(5)}$   &(-50, -22, 0)& 0.0458   &  $F^{(14)}$   &(-72, -22, 0) & 0.0417 \\
\hline $F^{(6)}$   &(-50, 22, 22)& 0.0292  &  $F^{(15)}$   &(-72, 22, 22) & 0.0333\\
\hline $F^{(7)}$   &(-50, 22, -22)& 0.025  &  $F^{(16)}$   &(-72, 22, -22)& 0.05  \\
\hline $F^{(8)}$   &(-50, -22, 22)&  0.0375  &  $F^{(17)}$   &(-72, -22, 22)& 0.0333 \\
\hline $F^{(9)}$   &(-50, -22, -22)& 0.0375   &  $F^{(18)}$   &(-72, -22, -22)& 0.0375\\
\hline
\end{tabular}
\caption{Data for the voids $F^{(j)}$, $j=1, \dots, 18$.}
\label{Table1a}
\end{table}

Figure \ref{fig1} shows the asymptotic solution $u_N$ of the mixed boundary value problem (part (b) of the figure) and its numerical counterpart obtained in COMSOL 3.5 (part (a) of the figure). This computation  has been produced for a spherical body containing $18$ small voids defined in Table \ref{Table1a}. 
The relative error for the chosen configuration
does not exceed $2 \% ,$ which confirms a very good agreement between the asymptotic and numerical results, which are
visually indistinguishable in Fig. \ref{fig1}a and Fig. \ref{fig1}b.


\begin{figure}[h]
\begin{center}
\includegraphics[scale=0.5]{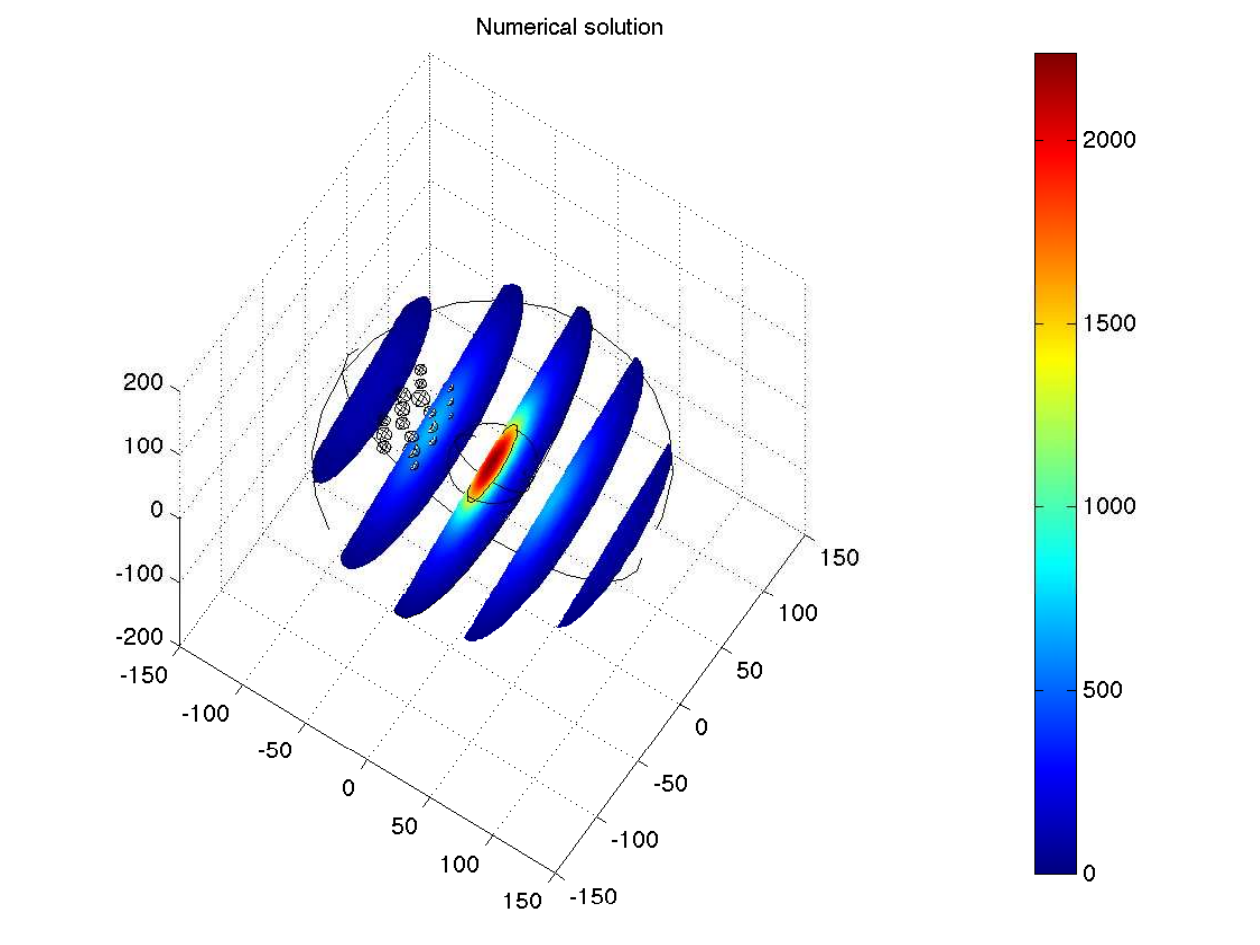} ~~~~\includegraphics[scale=0.5]{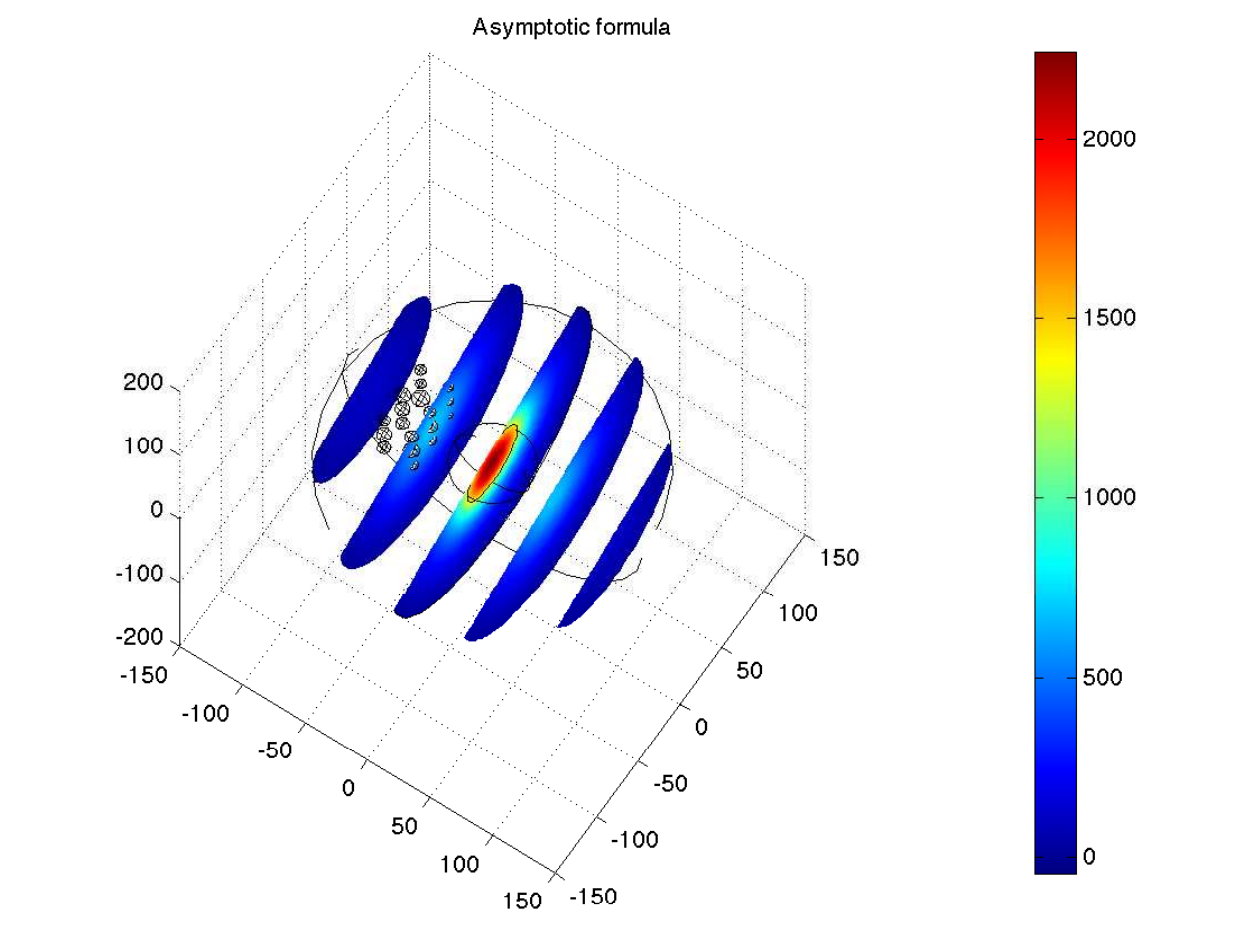}
\end{center}
\centerline{(a)   ~~~~~~~~~~~~~~~~~~~~~~~~~~~~~~~~~~~~~~~~~~~~~~~~~~~~~~~~~~~~  (b)}
\caption{Perforated domain containing 
$18$ holes: (a) Numerical solutions produced in COMSOL; (b) Asymptotic approximation.} \label{fig1}
\end{figure}

The computation was performed on Apple Mac, with 4Gb of RAM, and the number $N=18$ was chosen because any further increase in the number of voids resulted in a large three-dimensional computation, which exceeded the amount of available memory.
Although, increase in RAM can allow for a larger computation, it is evident that  three-dimensional finite element computations for a mesoscale geometry have serious limitations. On the other hand, the analytical asymptotic formula can still be used on the same computer for sigtnificantly larger number of voids.

In the next subsection, we show such an example where the number of voids within the mesoscale cloud runs upto $N=1000$, which would simply be unachievable in a finite element computation in COMSOL 3.5 with the same amount of RAM available.

\subsection{Non-uniform cloud containing a large number of spherical voids}

Here we consider the same mixed boundary value problem as in Section \ref{cloud_example}, but the cloud of voids is chosen in such a way that the number $N$ may be large and voids of different radii are distributed in a non-uniform arrangement.
For different values of $N$, the overall volume of voids is preserved - examples of the clouds used here are shown in Fig.
\ref{diagramcloud}.

The results  are based on the numerical implementation of formula (\ref{introeq1a}) in MATLAB.

%
The cloud $\omega$ 
is assumed to be the cube with side length $\frac{1}{\sqrt{3}}$ and the centre at $(3, 0, 0)$. 
Positioning of voids is described as follows. Assume we have $N=m^3$ voids, where $m=2,3, \dots$.
 Then $\omega$ is divided into   
 $N$ smaller cubes of side length $h=\frac{1}{\sqrt{3}m}$, and the centres of voids are placed at
\[\BO^{(p,{q,r})}=\Big(3-\frac{1}{2\sqrt{3}}+\frac{2p-1}{2}h,-\frac{1}{2\sqrt{3}}+\frac{2q-1}{2}h, -\frac{1}{2\sqrt{3}}+\frac{2r-1}{2}h\Big) \]
for $p, q, r=1, \dots, m$, and we assign their radii $\rho_{p,{q,r}}$ by
\[\rho_{p,{q,r}}=\left\{\begin{array}{ll}
\displaystyle{\frac{h}{5}} &\quad \text{if }{p>q}\;,\\
\displaystyle{\frac{\alpha h}{2}}& \quad \text{if } p<q\;,\\
\displaystyle{\frac{h}{4} }& \quad \text{if } p=q\;,
\end{array}\right.\]
where $\alpha <1$, 
and it is chosen in such a way that the overall volume of all voids within the cloud remains constant for different $N$.
An elementary calculation suggests that there will be $m^2$ voids with radius $\frac{h}{4}$ and equal number $\frac{m^2(m-1)}{2}$ of voids with radius $\frac{h}{5}$ or $\frac{\Ga h}{2}$.

Assuming that the volume fraction of all voids within the cube is equal to $\beta$, we have
\[\frac{4\pi h^3}{3} \Big(\frac{m^2(m-1)(8+125 \alpha^3)}{2000}+\frac{m^2}{64}\Big)=\beta \frac{1}{3\sqrt{3}}\;,\]
and hence
\begin{equation}\label{alphaeq}
\alpha^3=\frac{16m}{m-1}\Big\{\frac{3}{4\pi}\beta -\frac{125+32(m-1)}{8000 m}\Big\}\;.
\end{equation}
In particular, if $N \to \infty$, the limit value $\Ga_\infty$ becomes
\begin{equation}\label{alphainf}
\alpha_\infty=\Big\{\frac{12}{\pi}\beta -\frac{8}{125}\Big\}^{1/3}\;.
\end{equation}
In the numerical computation of this section, $\beta = \pi/25.$

Taking $R=7$ and $\rho=2$, we 
compute the leading order approximation of $u_N-v$, as defined  in the  asymptotic formula (\ref{introeq1a}), along the line $\gamma$ at the intersection of the planes $x_2=-1/(2\sqrt{3})$ and $x_3=-1/(2\sqrt{3})$, for $N=8, 125, 1000$. Fig. \ref{fig4} below shows the configuration of the cloud of voids for a) $N=8$ and b) $N=125$. For a large number of voids $(N=1000)$,  Fig. \ref{fig4a}a) shows the cloud and Fig \ref{fig4a}b) includes the graph of $\alpha$ versus $N$. The plot of $u_N-v$ given by (\ref{introeq1a}) for $2\le x_1\le 4$ is shown in Fig. \ref{fig5}. The asymptotic correction has been computed along the straight line $\gamma=\{ x_1 \in \mathbb{R}, x_2=-1/(2\sqrt{3}),  x_3=-1/(2\sqrt{3})\}$. Dipole type fluctuations are clearly visible on the diagram. Beyond $N=1000$ the graphs are visually indistinguishable and hence the values $N=8, 125, 1000$, as in Figures \ref{fig4} and \ref{fig4a} have been chosen in the computations. The algorithm is fast and does not impose periodicity constraints on the array of small voids.
   \begin{figure}[htbp]
\begin{minipage}[b]{0.495\linewidth}
\centering
\includegraphics[width=\textwidth]{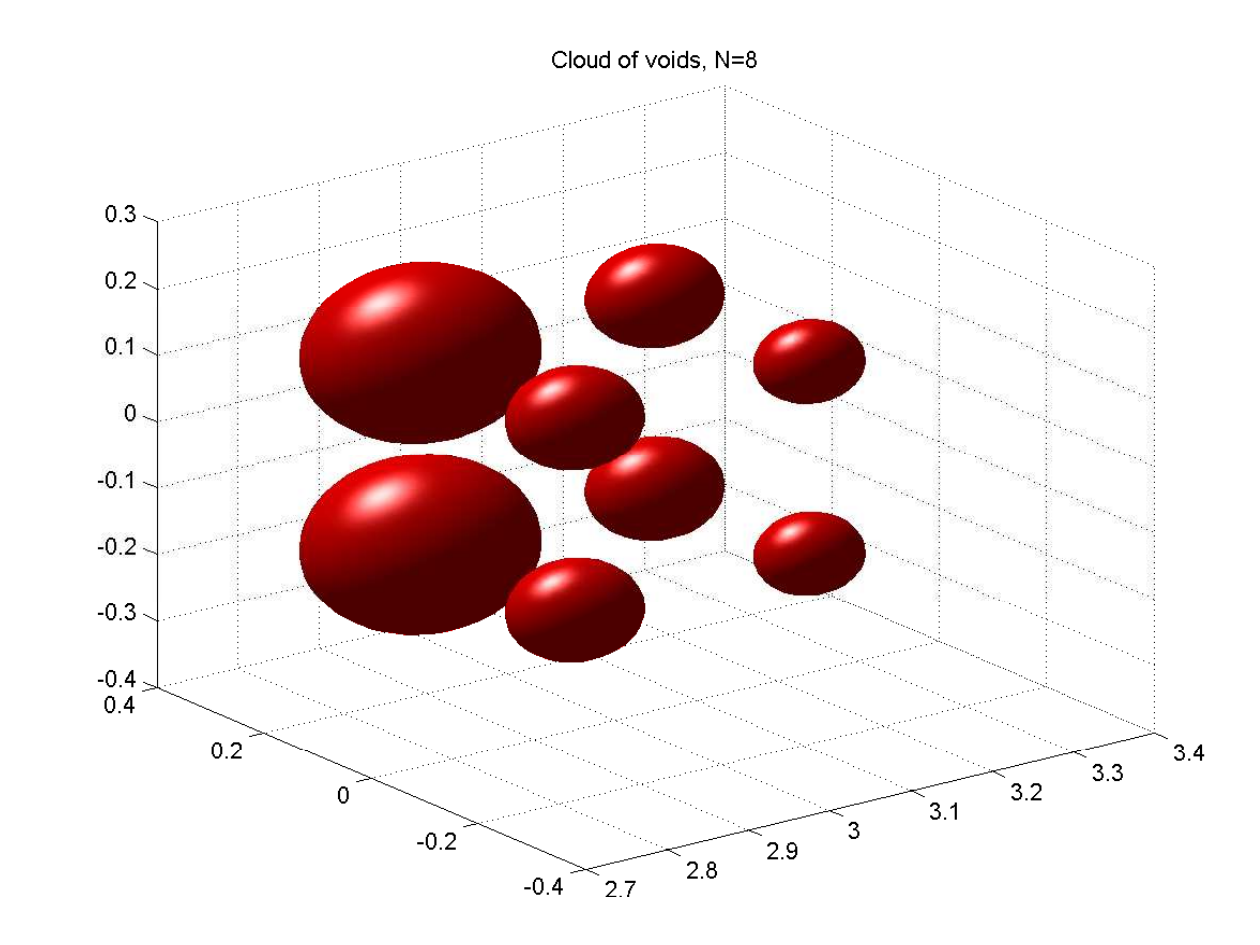}

a)
\end{minipage}
\begin{minipage}[b]{0.495\linewidth}
\centering
\includegraphics[width=\textwidth]{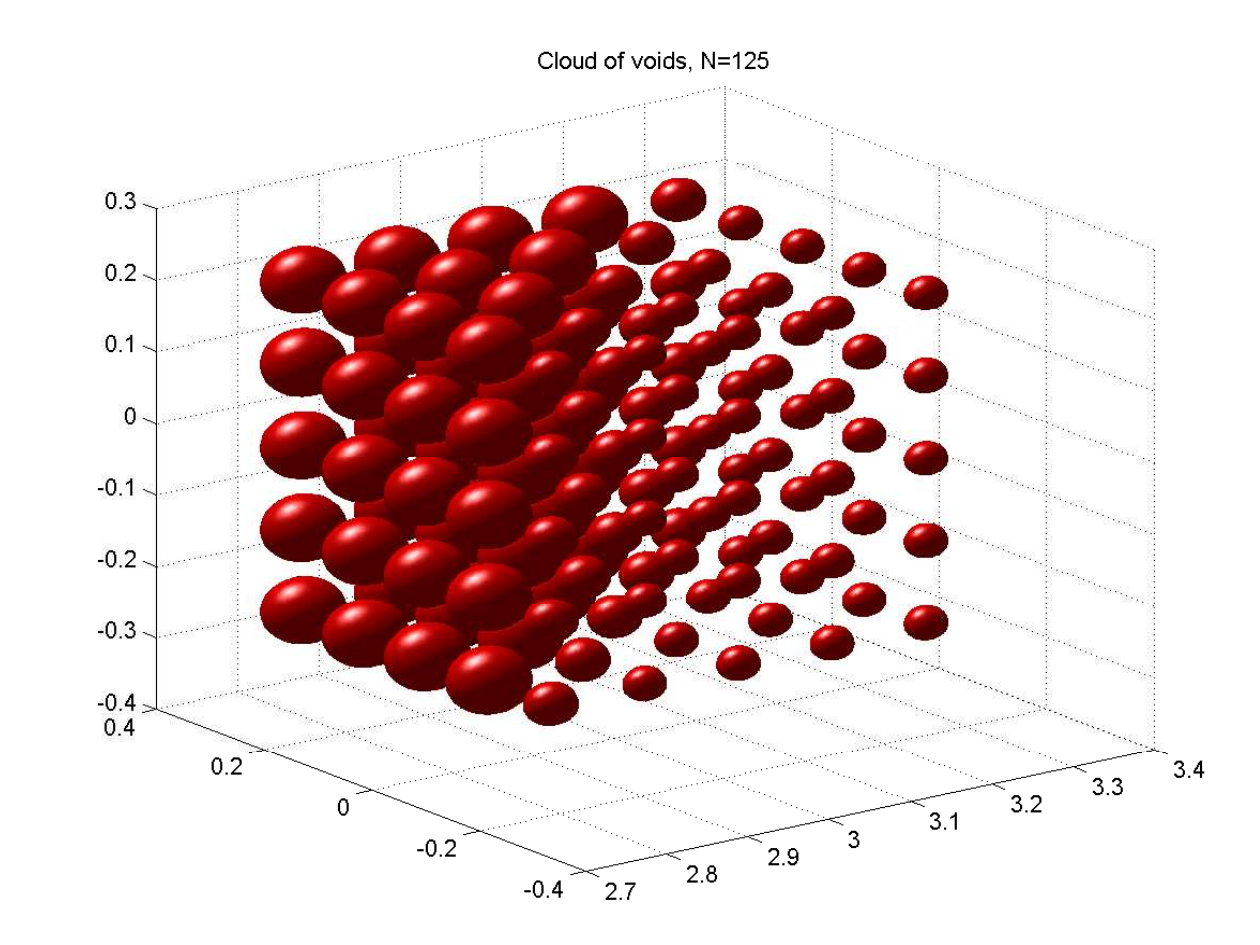}

b)
\end{minipage}
\caption{The cloud of voids for the cases when a) $N=8$ and b) $N=125$.}
\label{fig4}
\end{figure}

 \begin{figure}[htbp]
\begin{minipage}[b]{0.495\linewidth}
\centering
\includegraphics[width=\textwidth]{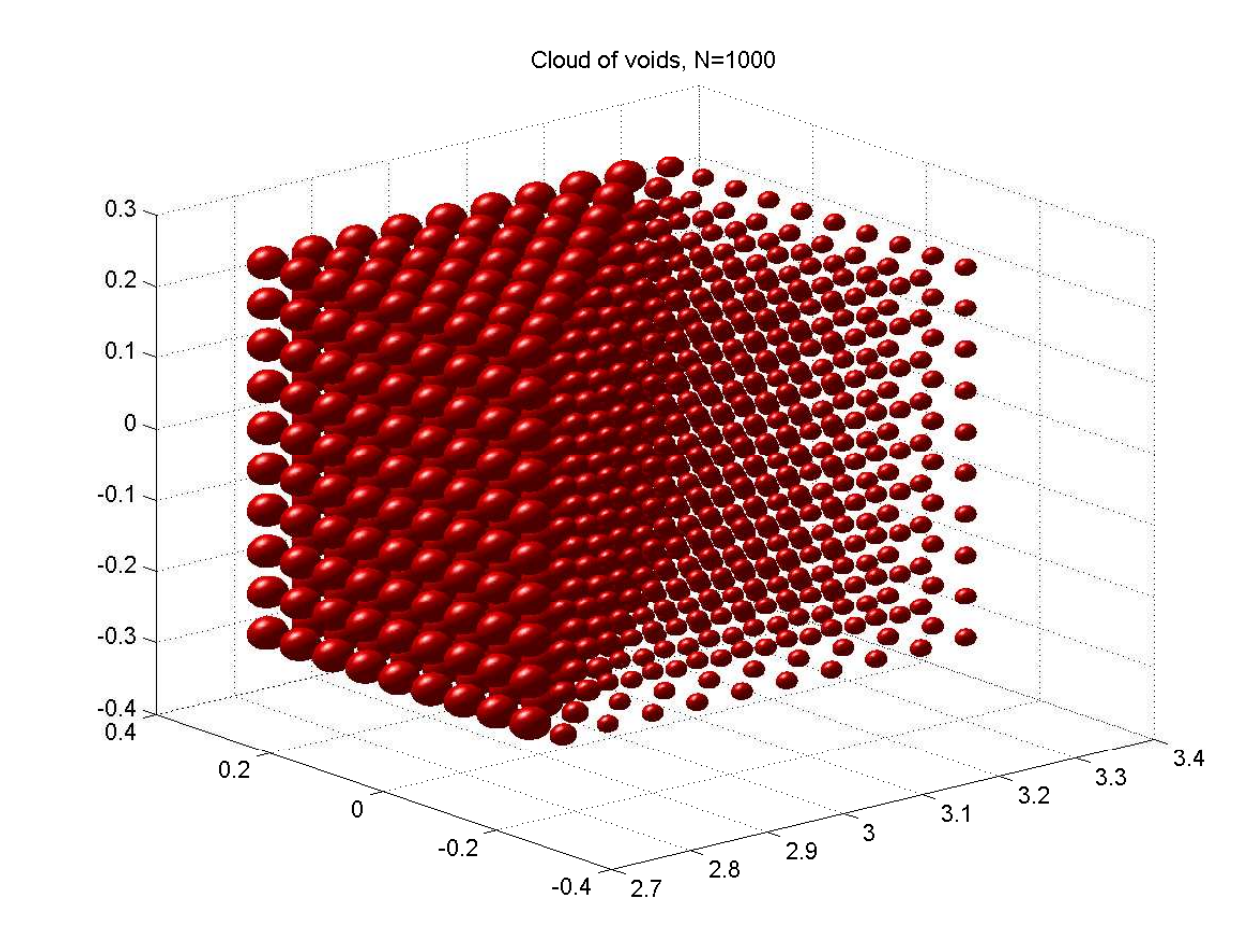}

a)
\end{minipage}
\begin{minipage}[b]{0.495\linewidth}
\centering
\includegraphics[width=\textwidth]{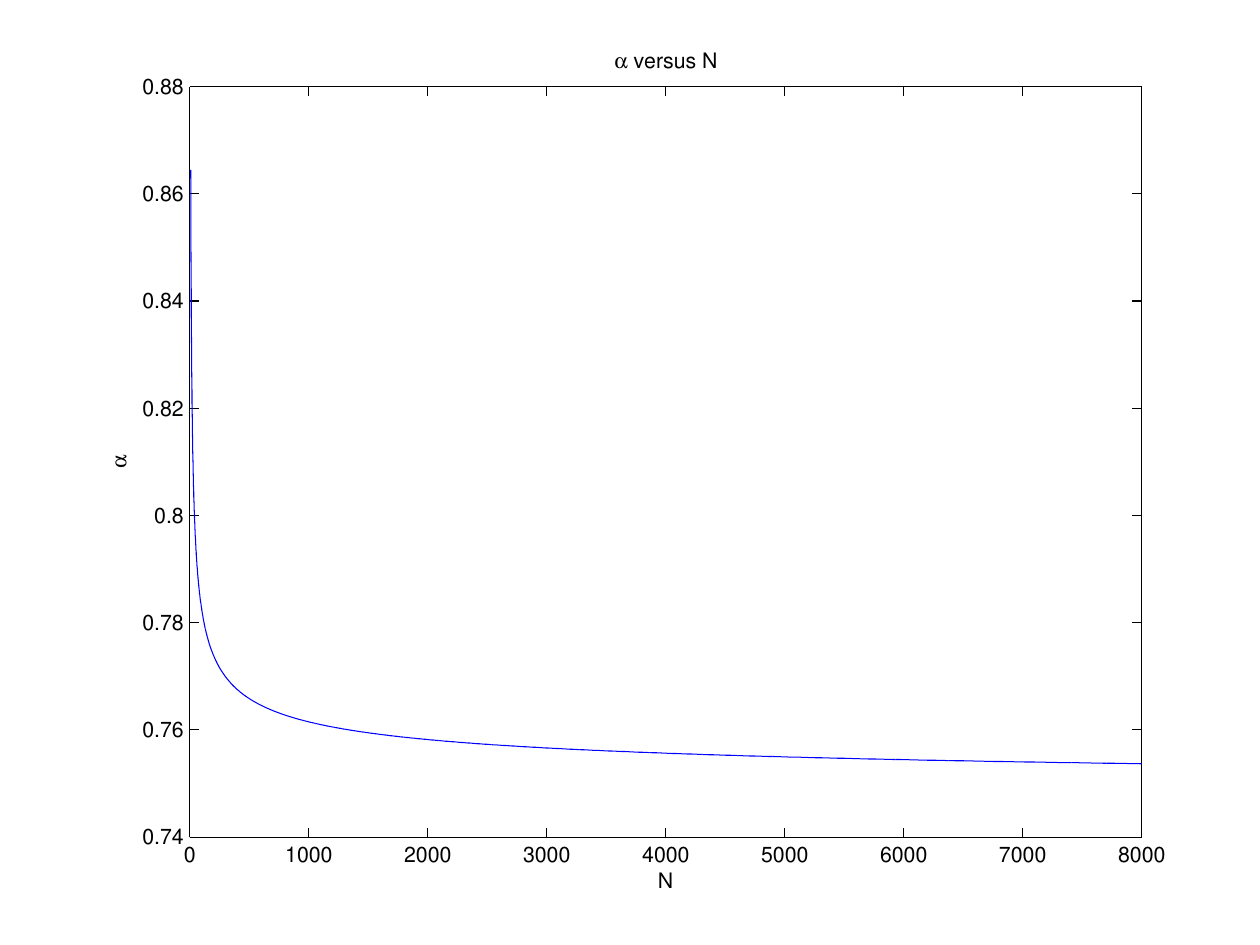}

b)
\end{minipage}
\caption{a) The cloud of voids for the cases when $N=1000$, b) The graph of $\alpha$ versus $N$ given by formula (\ref{alphaeq}) when $\beta= \pi/25$, for large $N$ we see that $\alpha$ tends to 0.7465 which is predicted value present in (\ref{alphainf}).}
\label{fig4a}
\end{figure}

 \begin{figure}[htbp]
\centering
\includegraphics[width=0.7\textwidth]{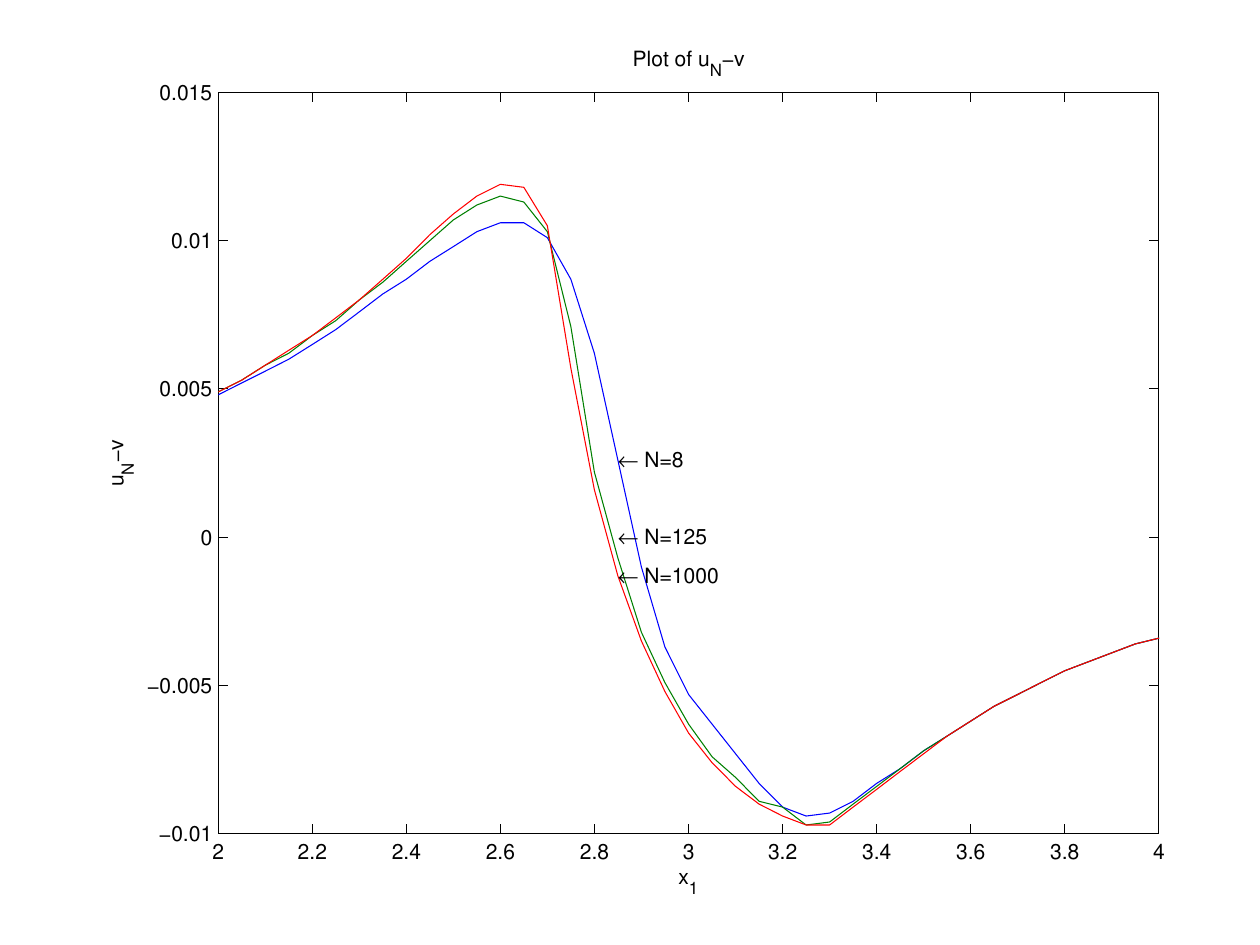}
\caption{The graph of $u_N-v$ given by (\ref{introeq1a}), for $2 \le x_1\le 4$ plotted along the straight line $\gamma$ adjacent to the cloud of small voids.}
\label{fig5}
\end{figure}

\vspace{0.1in}\emph{Acknowledgments. }We would like to acknowledge the financial support of the U.K. Engineering and Physical Sciences Research Council through the research grant EP/F005563/1.

\end{document}